\documentclass[a4paper,12pt]{article}

\usepackage{amsthm,amsmath,amsfonts,xcolor,graphicx}
\usepackage{thm-restate}
\usepackage{mathabx}
\usepackage{hyperref}
\usepackage[noabbrev,capitalize]{cleveref}

\numberwithin{equation}{section}
\newtheorem{theorem}{Theorem}[section]
\newtheorem{proposition}[theorem]{Proposition}
\newtheorem{lemma}[theorem]{Lemma}

\crefname{claim}{Claim}{Claims}

\newtheorem{corollary}[theorem]{Corollary}
\newtheorem{conjecture}[theorem]{Conjecture}
\newtheorem*{question*}{Question}

\theoremstyle{definition}
\newtheorem{definition}[theorem]{Definition}

\newtheorem{question}[theorem]{Question}
\newtheorem*{definition*}{Definition}

\theoremstyle{remark}

\usepackage{xspace}

\usepackage[textsize=tiny]{todonotes}
\usepackage{mathdots}

\newcommand{\fA}{\mathcal{A}}
\newcommand{\fB}{\mathcal{B}}

\newcommand{\fE}{\mathcal{E}}
\newcommand{\fF}{\mathcal{F}}
\newcommand{\fG}{\mathcal{G}}

\newcommand{\fI}{\mathcal{I}}

\newcommand{\fN}{\mathcal{N}}
\newcommand{\fO}{\mathcal{O}}
\newcommand{\fP}{\mathcal{P}}

\newcommand{\fR}{\mathcal{R}}

\newcommand{\fT}{\mathcal{T}}

\DeclareMathOperator{\re}{\mathcal R}
\DeclareMathOperator{\rere}{\overline{\mathcal R}}
\newcommand{\noco}{\fN_{\Pi}}
\newcommand{\edco}{\fE_{\Pi}}
\newcommand{\gee}{g_{\Pi}}

\newcommand{\rnoco}{\fN_{\re(\Pi)}}
\newcommand{\redco}{\fE_{\re(\Pi)}}
\newcommand{\rgee}{g_{\re(\Pi)}}
\newcommand{\rrnoco}{\fN_{\rere(\re(\Pi))}}
\newcommand{\rredco}{\fE_{\rere(\re(\Pi))}}
\newcommand{\rrgee}{g_{\rere(\re(\Pi))}}
\newcommand{\orrnoco}{\fN_{\rere(\Pi)}}
\newcommand{\orredco}{\fE_{\rere(\Pi)}}
\newcommand{\orrgee}{g_{\rere(\Pi)}}

\newcommand{\xA}{\mathsf{A}}
\newcommand{\xB}{\mathsf{B}}

\newcommand{\xa}{\mathsf{a}}
\newcommand{\xb}{\mathsf{b}}

\newcommand{\inn}{{\operatorname{in}}}
\newcommand{\out}{{\operatorname{out}}}

\newcommand{\sinn}{\Sigma_{\inn}}
\newcommand{\sout}{\Sigma_{\out}}
\newcommand{\spinn}{\Sigma^{\Pi}_{\inn}}
\newcommand{\spout}{\Sigma^{\Pi}_{\out}}
\newcommand{\rspinn}{\Sigma^{\re(\Pi)}_{\inn}}
\newcommand{\rspout}{\Sigma^{\re(\Pi)}_{\out}}

\newcommand{\rrspout}{\Sigma^{\rere(\re(\Pi))}_{\out}}
\newcommand{\orrspinn}{\Sigma^{\rere(\Pi)}_{\inn}}
\newcommand{\orrspout}{\Sigma^{\rere(\Pi)}_{\out}}

\newcommand{\congest}{$\mathsf{CONGEST}$\xspace}
\newcommand{\local}{$\mathsf{LOCAL}$\xspace}

\newcommand{\prodlocal}{$\mathsf{PROD\text{-}LOCAL}$\xspace}

\newcommand{\lca}{$\mathsf{LCA}$\xspace}
\newcommand{\volume}{$\mathsf{VOLUME}$\xspace}

\newcommand{\poly}{\operatorname{\text{{\rm poly}}}}

\newcommand{\N}{\mathbb{N}}

\begin{document}

\title{The Landscape of Distributed Complexities on Trees and Beyond} 





\author{Christoph Grunau\footnote{The author ordering was randomized. } \\ \small{ETH Zurich} \\ \small{cgrunau@inf.ethz.ch} 
\and \textcircled{r} \and
\and Václav Rozhoň\\ \small{ETH Zurich} \\ \small{rozhonv@ethz.ch}
\and \textcircled{r} \and
Sebastian Brandt\\ \small{CISPA Helmholtz Center for Information Security} \\ \small{brandt@cispa.de}
}

\maketitle

\begin{abstract}
    We study the local complexity landscape of locally checkable labeling (LCL) problems on constant-degree graphs with a focus on complexities below $\log^* n$. 
    Our contribution is threefold:
    \begin{enumerate}
        \item Our main contribution is that we complete the classification of the complexity landscape of LCL problems on trees in the \local model, by proving that every LCL problem with local complexity $o(\log^* n)$ has actually complexity $O(1)$. This result improves upon the previous speedup result from $o(\log \log^* n)$ to $O(1)$ by [Chang, Pettie, FOCS 2017].
        \item In the related \lca and \volume models [Alon, Rubinfeld, Vardi, Xie, SODA 2012, Rubinfeld, Tamir, Vardi, Xie, 2011, Rosenbaum, Suomela, PODC 2020], we prove the same speedup from $o(\log^* n)$ to $O(1)$ for \emph{all} constant-degree graphs. 
        \item Similarly, we complete the classification of the \local complexity landscape of oriented $d$-dimensional grids by proving that any LCL problem with local complexity $o(\log^* n)$ has actually complexity $O(1)$. This improves upon the previous speed-up from $o(\sqrt[d]{\log^* n})$ by Suomela, explained in [Chang, Pettie, FOCS 2017].
    \end{enumerate}
\end{abstract}

\section{Introduction}
\label{sec:intro}
One fundamental aspect in the theory of distributed computing is \emph{locality}: 
for a given graph problem, how far does each individual node has to see in order to find its own part of a globally correct solution.
The notion of locality is captured by the local complexity of the \local model of distributed computing \cite{linial1987LOCAL}.
In recent years, our understanding of locality has improved dramatically by studying it from a complexity theoretical perspective \cite{naorstockmeyer, chang2016exp_separation, chang2017time_hierarchy, brandt_grids, balliurooted21}. 
That is, characterizing the local complexity of all problems from a certain class of problems at once.
Essentially all the past research has focused on the study of so-called Locally Checkable Labeling problems (LCLs) \cite{naorstockmeyer}. 
Informally speaking, LCLs are graph problems that can be verified efficiently. That is, if a given output looks correct in the local neighborhood of each node, then the output is guaranteed to be a globally correct solution.
Prominent examples of LCLs are Maximal Independent Set, Maximal Matching and all kinds of coloring problems.
LCLs are only defined for constant-degree graphs. 

Interestingly, the study of local complexity has also been shown to be quite relevant for other subfields of computer science such as local computation algorithms or massively parallel computation \cite{rubinfeld2011fast,ghaffari2019conditional_hardness}, and is finding its applications in subareas of discrete mathematics as well \cite{holroyd_schramm_wilson2017,bernshteyn2020LLL,brandt_chang_grebik_grunau_rozhon_vidnyanszky2022trees}.

\paragraph*{Below $\log^* n$}
A classical result of Linial \cite{linial1987LOCAL} shows that the deterministic \local complexity of $O(\Delta^2)$-coloring a graph with maximum degree $\Delta$ is $\Theta(\log^* n)$.\footnote{The function $\log^*(n)$ is defined as the minimum number of times one needs to repeatedly apply the $\log$ function to $n$ until the result becomes at most $1$.} In particular, this result directly implies that many basic symmetry breaking problems like ($\Delta+1$)-coloring have a deterministic \local complexity of $\Theta(\log^* n)$ on constant-degree graphs. In contrast, there are many basic problems with complexity $O(1)$. An example is to ``find the maximum degree of a node in your $2$-hop neighborhood''. The common theme for all our contributions is to understand what happens in between the two worlds. By a result of Chang and Pettie~\cite{chang2017time_hierarchy}, whose proof was inspired by ideas introduced by Naor and Stockmeyer~\cite{naorstockmeyer}, we know that all LCLs with local complexity $o(\log \log^* n)$ have, in fact, complexity $O(1)$. 

Complementing this result, the paper \cite{balliu2018new_classes-loglog*-log*} gave examples of LCLs with complexities between $\Theta(\log\log^* n)$ and $\Theta(\log^* n)$. 
Roughly speaking, examples of those problems are the following: one is asked to solve a basic $\Theta(\log^*n)$-complexity problem on a path. The input graph, however, is not just a path. It is a path $P$ plus a shortcutting structure on top of it that ensures that the $t$-hop neighborhood of each node $u \in P$ in the full input graph $G$ actually contains the $f(t)$-hop neighborhood of $u$ in the subgraph $P$ for some function $f$. 
This results in a problem with complexity equal to $\Theta(f^{-1}(\log^* n))$. 
The growth rate of the function $f$ can range from linear (if no shortcutting is involved) up to exponential, since the neighborhood at distance $t$ in the full graph $G$ can contain up to roughly $\Delta^t$ nodes. 
Hence we obtain complexities from the range $\Theta(\log \log^* n) - \Theta(\log^* n)$. 
It should be noted that such constructed problems have a special structure. In fact, one can observe the following two points:
\begin{enumerate}
    \item To construct such a problem, the input graph needs to contain shortcuts and, hence, cycles. In \cref{thm:main_informal} we prove that if the input graph is a tree, there cannot be any such problems. This completes the classification of LCLs on trees. 
    \item Although the necessary radius that a node $u$ needs to check to solve the constructed problem can be only $O(\log \log^* n)$, the number of nodes that $u$ needs to look at to solve the problem is still $O(\log^* n)$. 
    In \cref{thm:volume_speedup_informal} we prove that in the \volume complexity model, where the main measure is the number of nodes $u$ needs to query, not the radius, there are in fact no complexities between $\omega(1)$ and $o(\log^* n)$. 
\end{enumerate}

\subsection{Main Contribution: Finishing the Classification of LCLs on Trees}\label{subsec:introtree}
In this section we present a prominent achievement of the theory of LCLs: by a long line of work~\cite{balliu2020almost_global_problems,balliu2018new_classes-loglog*-log*,brandt_LLL_lower_bound,chang2016exp_separation,chang2017time_hierarchy,cole86,FischerGK17,ghaffari_grunau_rozhon2020improved_network_decomposition,ghaffari_harris_kuhn2018derandomizing,ghaffari_su2017distributed_degree_splitting,linial1987LOCAL,naor1991lower_bound_ring_coloring,naorstockmeyer,RozhonG19}, we now know that there are only four types of LCL problems, as described in \cref{fig:big_picture}, bottom left:
\begin{enumerate}
    \item[(A)] Problems that can be solved in $O(1)$ rounds.
    \item[(B)] Problems with complexity (both randomized and deterministic) in the range $\Theta(\log\log^* n) - \Theta(\log^* n)$; these include basic symmetry breaking problems like ($\Delta+1$)-coloring and maximal independent set.
    \item[(C)] Problems with randomized complexity $\poly\log\log n$ and deterministic complexity $\poly\log n$; these problems can be solved by reformulating them as an instance of the Lovász local lemma (LLL).
    \item[(D)] Global problems with complexity $\Omega(\log n)$; their randomized and deterministic complexity is the same up to a polylogarithmic factor.
\end{enumerate}


One particular focus in the long line of work on complexities of LCLs has been on understanding the complexity landscape on trees---not only because trees are a natural subclass of general graphs, but also because they turned out to be important for the development of understanding the complexities on general graphs; in particular many lower bounds for LCLs are proven on trees, e.g., via the round elimination technique \cite{brandt19automatic_speedup_theorem}.
As a result the complexity landscape on trees is quite well understood: previous to our work, each LCL on trees was known to have one of the following complexities (unless stated otherwise, the deterministic and randomized complexities are the same):
\begin{enumerate}
    \item $O(1)$,
    \item in the range $\Theta(\log \log^* n) - \Theta(\log^* n)$,
    \item deterministic complexity $\Theta(\log n)$ and randomized complexity $\Theta(\log \log n)$,
    \item complexity $\Theta(\log n)$,
    \item complexity $\Theta(n^{1/k})$ for some positive integer $k$.
\end{enumerate}

This follows from a long line of work~\cite{naorstockmeyer,brandt_LLL_lower_bound,chang2016exp_separation,chang2017time_hierarchy,chang2018edge_coloring,balliu2020almost_global_problems,chang2020n1k_speedups}.
The only missing piece was to determine for which complexities in the range $\Theta(\log \log^* n) - \Theta(\log^* n)$ there actually exists an LCL with that complexity.
We complete the complexity landscape of LCLs on trees by showing that any LCL with a complexity in this range has actually complexity $\Theta(\log^* n)$.
That is, we prove that any LCL on trees with complexity $o(\log^* n)$ has complexity $O(1)$.

\begin{theorem}[Informal version of \cref{thm:treemain}]
\label{thm:main_informal}
Let $\Delta$ be any fixed positive integer. Any LCL on trees (with maximum degree at most $\Delta$) with \local complexity $o(\log^* n)$ has, in fact, \local complexity $O(1)$.  
\end{theorem}

Put together with the aforementioned work on the classification of LCLs on trees, we get the following corollary (see \cref{fig:big_picture}, top left, for an illustration). 

\begin{corollary}[Classification of LCLs on trees, see \cref{fig:big_picture}]
\label{cor:speedup_trees}
Let $\Pi$ be an LCL on trees. Then the deterministic/randomized \local complexity of $\Pi$ is one of the following:
\begin{enumerate}
    \item $O(1)$,
    \item $\Theta(\log^*n)$,
    \item $\Theta(\log n)$ deterministic and $\Theta(\log\log n)$ randomized,
    \item $\Theta(\log n)$,
    \item $\Theta(n^{1/k})$ for $k \in \mathbb{N}$.
\end{enumerate}
Moreover, all mentioned complexity classes are non-empty, i.e., for each class (and each $k$), there exists some LCL with the indicated complexity.
\end{corollary}

An intriguing implication of our work is that the complexity landscape of LCLs on trees is \emph{discrete}, while the landscape of LCLs on general (constant-degree) graphs contains ``dense'' regions (in classes (B) and (D)~\cite{balliu2018new_classes-loglog*-log*,balliu2020almost_global_problems}).
It would be interesting to see to what extent this is a more general phenomenon beyond LCLs on constant-degree graphs.

We also note that for any LCL problem $\Pi$, its complexity $T$ on trees is the same as its complexity on high-girth graphs (assuming that the girths are at least a suitable additive constant larger than $2T$).
The reason is that, due to the definition of an LCL problem, the correctness of a claimed solution only depends on the correctness of the solution in each node neighborhood of radius $T$ plus some constant.
We conclude that our gap result also applies to high-girth graphs (roughly speaking with girths from $\omega(\log^* n)$).

\begin{figure}
    \centering
    \includegraphics[width=.48\textwidth]{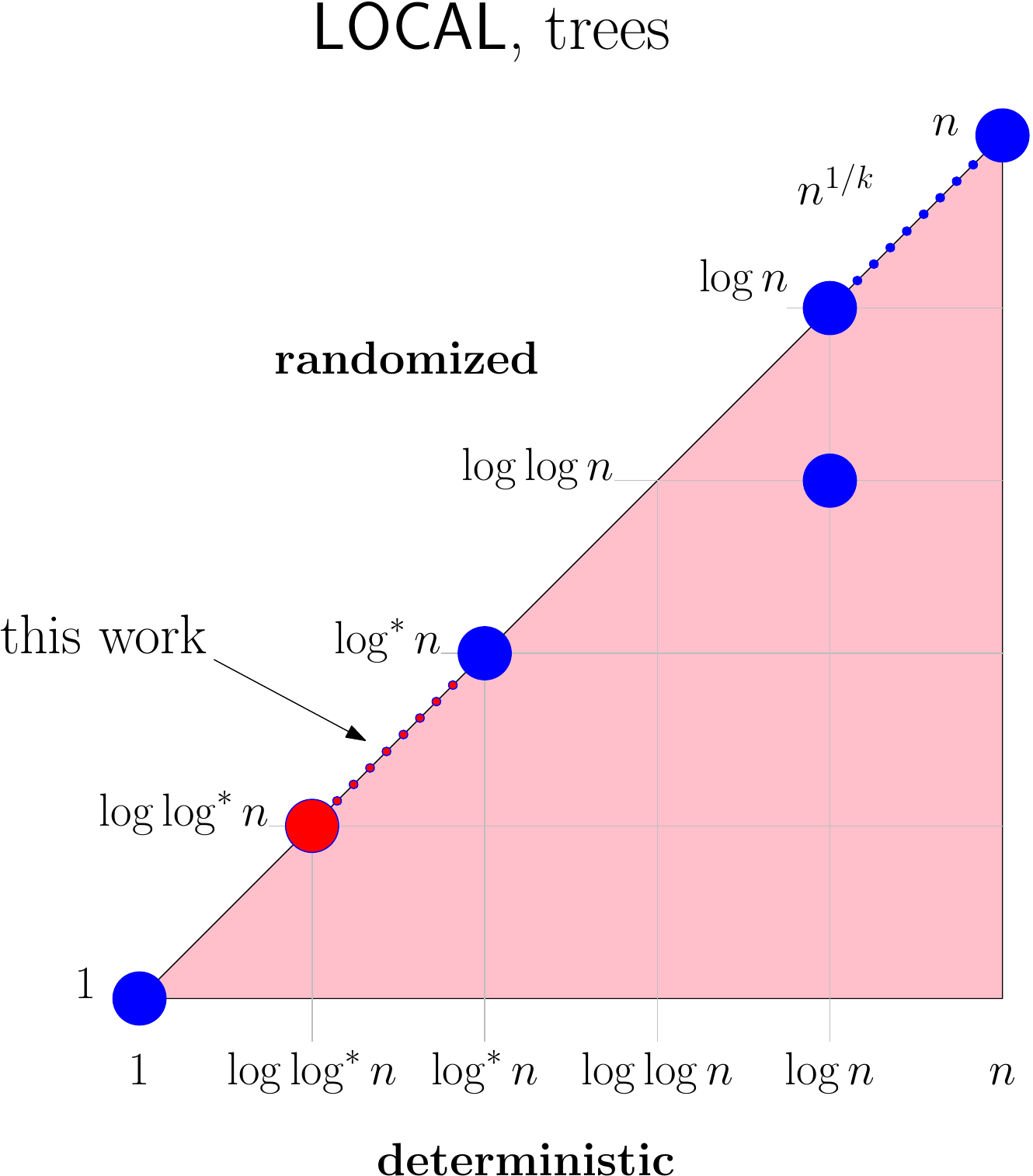}
    \includegraphics[width=.47\textwidth]{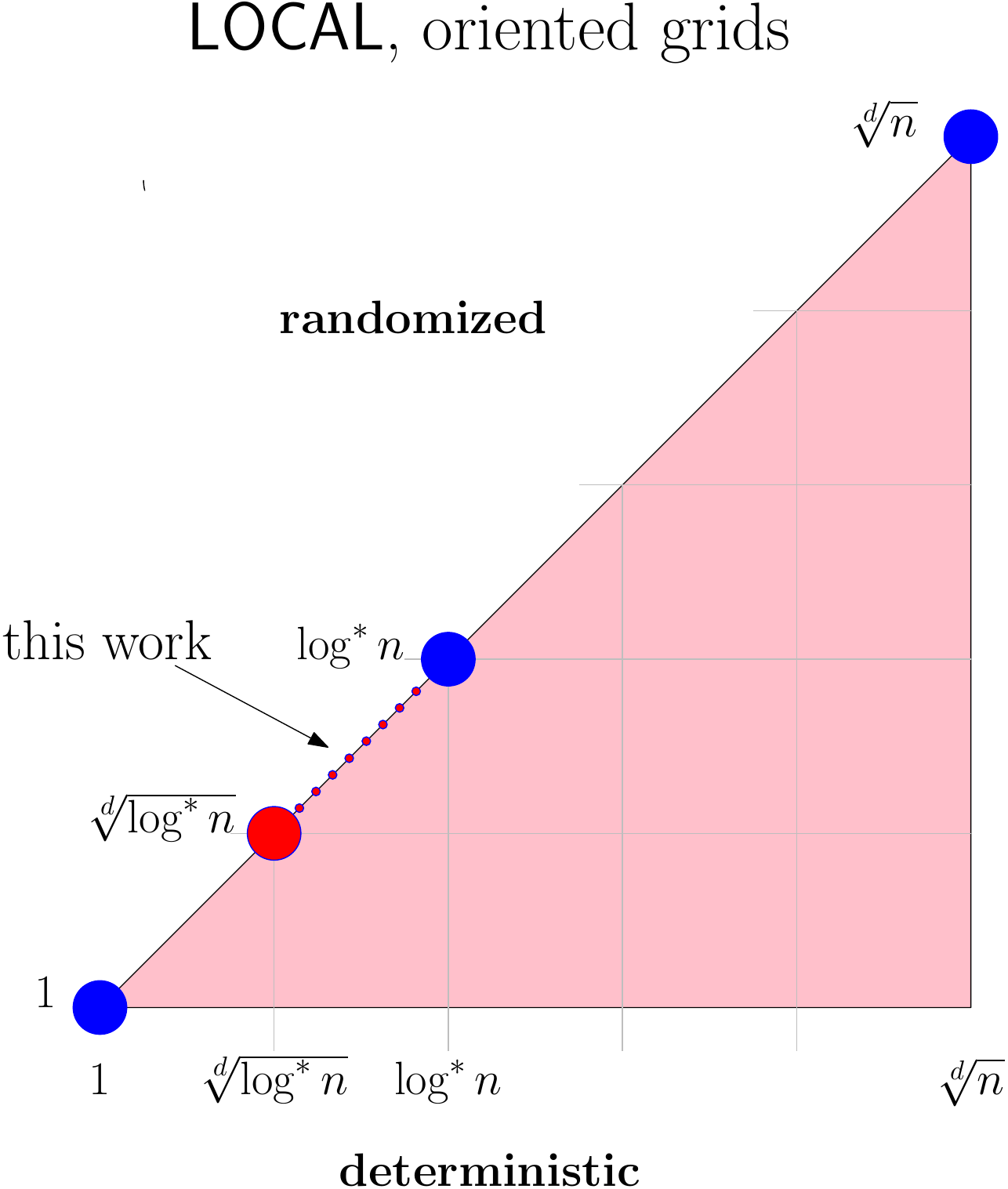}
    
    \vspace*{2cm}
    
    \includegraphics[width=.48\textwidth]{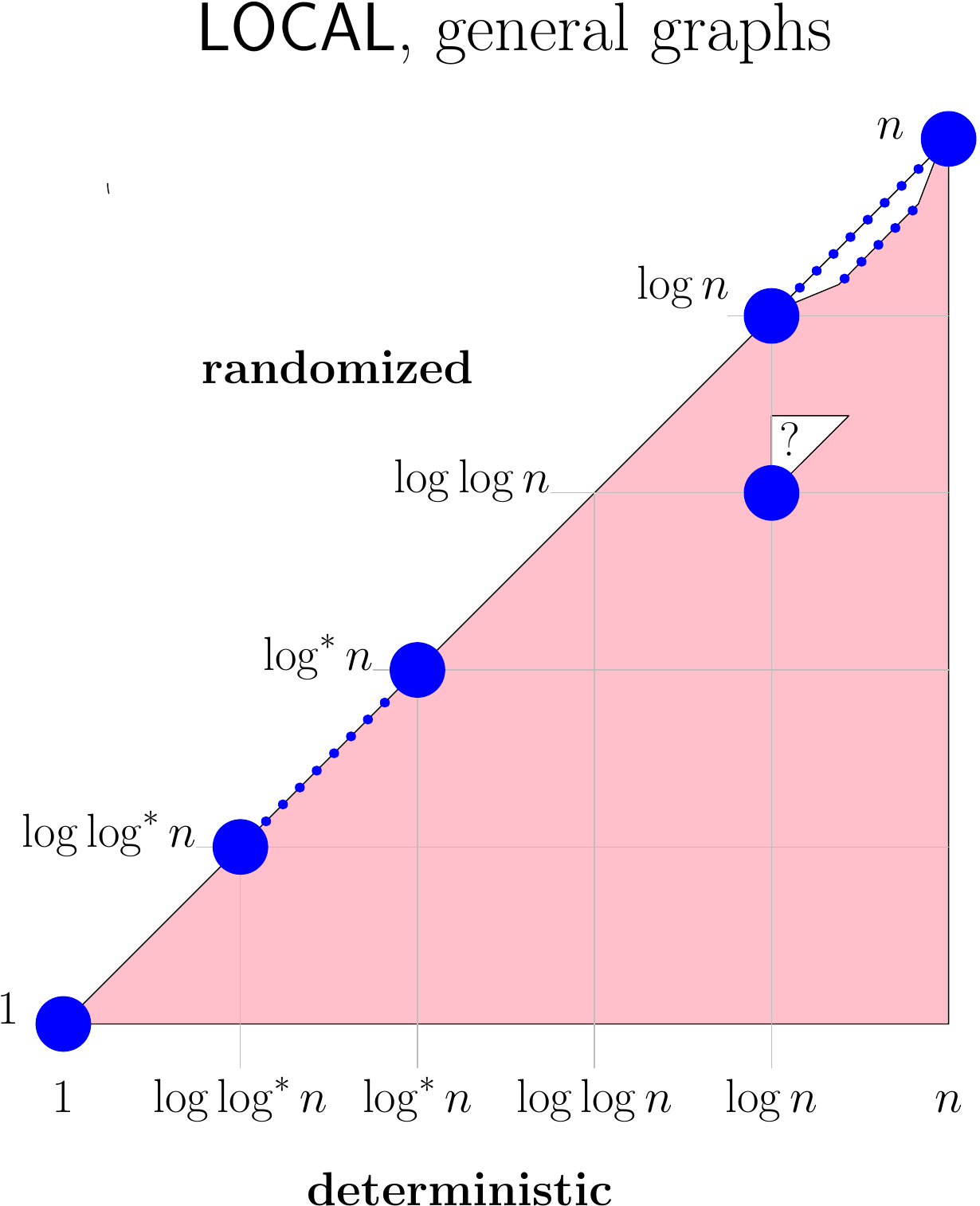}
    \includegraphics[width=.48\textwidth]{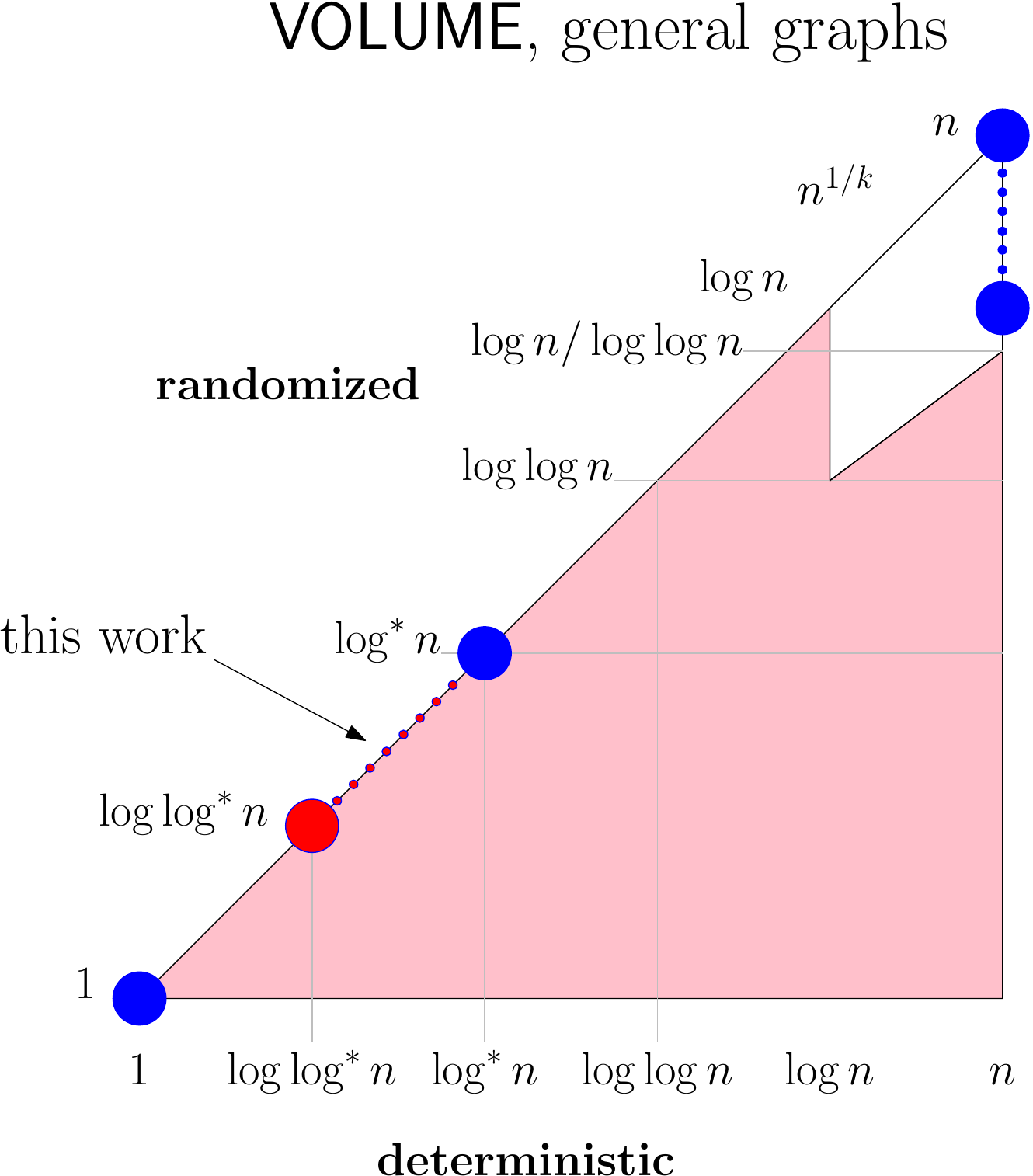}
    
    \caption{The LCL landscape of local complexities on trees (top left), oriented grids (top right), general constant-degree graphs (bottom left) and the landscape of the volume model (bottom right). 
    Blue circles correspond to possible deterministic and randomized local complexities, while the red color means that no problems with the given complexity are possible in the area. The arrows point to our contribution: completion of the classification of LCLs on trees and oriented grids by showing that there are no problems with complexities between $\Theta(\log^* n)$ and $O(1)$. We now also understand the landscape of the \volume / \lca model in the regime of low complexities. 
    See \cite{Jukka_presentation} for a friendly introduction to the topic of LCL landscapes. }
    \label{fig:big_picture}
\end{figure}

\paragraph*{Our method, in a nutshell}

Our approach is based on the round elimination technique~\cite{brandt19automatic_speedup_theorem}, a highly successful technique for proving local lower bounds. 
In essence, round elimination is an explicit process that takes an LCL $\Pi$ on trees as input and returns an LCL $\Pi'$ with complexity exactly one round less. More precisely:
\begin{enumerate}
    \item[(1)] 
    If there is a $T$-round randomized algorithm $\fA$ for $\Pi$, then there is also a $(T-1)$-round randomized algorithm $\fA'$ for $\Pi'$ such that the (local) failure probability of $\fA'$ is bounded by a reasonable function of the (local) failure probability of $\fA$. 
    \item[(2)] If we have a $(T-1)$-round algorithm $\fA'$ for $\Pi'$, we can use it to construct a $T$-round algorithm $\fA$ for the original problem $\Pi$; if $\fA'$ is deterministic, then $\fA$ is deterministic as well. 
\end{enumerate}

So far, in the literature, the standard use case for applying round elimination has been to prove lower bounds for some \emph{concrete, fixed} problem such as maximal matching or Lov\'asz local lemma~\cite{brandt_LLL_lower_bound,brandt19automatic_speedup_theorem,balliu2019LB,balliu2020ruling}.
We provide a novel application of round elimination by showing that, perhaps surprisingly, it can also be used to prove gap results, which are results that reason about \emph{all} LCLs on a given graph class. 
More precisely, we show that with the tool of round elimination at hand, there is an elegant way to prove Theorem~\ref{thm:main_informal}, which roughly proceeds as follows.

We start with any problem $\Pi$ for which there exists a randomized algorithm $\fA$ that solves $\Pi$ in $T(n) = o(\log^* n)$ rounds, with probability $1 - 1/\poly(n)$. We fix some sufficiently large number $n_0$ of nodes, and apply bullet point (1) $T = T(n_0)$ times to get a $0$-round algorithm $\fA^{(T)}$ for a certain problem $\Pi^{(T)}$. By analyzing the development of the (local) failure probabilites of the algorithms appearing during the $T$ applications of (1), we can show that algorithm $\fA^{(T)}$ still has a large probability of success, and the fact that $\fA^{(T)}$ is a $0$-round algorithm enables us to infer that $\Pi^{(T)}$ is in fact so easy that it can be solved with a \emph{deterministic} $0$-round algorithm. 
Finally, we apply bullet point (2) $T$ times to obtain a deterministic $T$-round algorithm for the original problem $\Pi$.
Due to fixing the number of nodes to $n_0$, the obtained $T$-round algorithm is only guaranteed to produce a correct output on $n_0$-node trees; however, due to the nature of $0$-round algorithms and the precise definition of the round elimination process (which are both independent of the number of nodes of the input graph), the obtained algorithm can be shown to also work on trees with an arbitrary number of nodes, with precisely the same, constant runtime $ T(n_0) = O(1)$.

Unfortunately, the known approach for analyzing the change of (local) failure probability in bullet point (1) considers only the restricted setting of regular graphs and LCLs without inputs\footnote{We say that an LCL is an \emph{LCL without inputs} if the correctness of a solution does not depend on input labels in the graph (such as lists in a list coloring problem). In the general setting, an LCL allows the correctness to depend on input labels (though we might emphasize this by using the term \emph{LCL with inputs}).} (which usually suffices when proving a lower bound for a concrete LCL).
One of our technical contributions is to develop an extension that also works in the general setting of irregular graphs and LCLs with inputs, which might be of independent interest.

\paragraph*{Further related work}
Previous to our work, the gap result of \cref{thm:main_informal} was known for a subclass of LCLs, called homogeneous LCLs \cite{balliu2019hardness_homogeneous}.
Roughly speaking, problems in this class require the output of a node $u$ to be correct only if the part of the tree around $u$ is a perfect $\Delta$-regular tree without any inputs.
In contrast, our argument works in full generality, i.e., the degrees of nodes can differ and there can be inputs in the graph. Also, we believe that our argument is substantially simpler conceptually.

A recent work~\cite{balliurooted21} gave a complete classification of possible complexities of LCLs on \emph{rooted regular} trees, showing that each LCL on such trees has a complexity of $O(1)$, $\Theta(\log^* n)$, $\Theta(\log n)$, or $\Theta(n^{1/k})$ for some positive integer $k$ (and all of these complexity classes are nonempty).
Moreover, the complexity of each LCL is independent of whether randomization is allowed and whether the \local or the \congest model\footnote{The \congest model differs from the \local model in that it only allows messages of size $O(\log n)$ bits.} is considered.
We note that their approach relies heavily on the provided orientation and, to the best of our knowledge, does not generalize to unrooted trees.
We will discuss the decidability results from this work in Section~\ref{sec:gendec}.

Even more recently, a paper \cite{balliusmall21} showed that the asymptotic complexity of any LCL problem on trees is the same in the \local and the \congest model.
This implies that the complexity landscape of LCLs on trees in \congest is precisely the same as in \local, and in particular extends our gap between $\omega(1)$ and $o(\log^* n)$ to the \congest model.
They also show that their result does not extend to general graphs by explicitly providing an LCL with different complexities in \local and \congest.

\subsection{Speedup in the \volume Model}

Recently, Rosenbaum and Suomela\cite{Rosenbaum2020} initiated the study of the complexity landscape of locally checkable problems in the \volume model on general constant-degree graphs, as was done before in the \local model. 
In the \volume model, a node $v$ can adaptively probe its local neighborhood in order to determine its local output.
In comparison to the \local model, $v$ does not learn its entire local neighborhood. Instead, it can only explore parts of it before computing its output.
The \volume model is very similar (and, in the complexity regime we consider, identical, due to the work of \cite{goos_nonlocal_probes_16}) to the well-studied \lca model \cite{alon2012space,rubinfeld2011fast}. 
The basic landscape of LCLs in the \volume model following from \cite{Rosenbaum2020,brandt_grunau_rozhon2021LLL_in_LCA} is summarized in the bottom right part of \cref{fig:big_picture}. Note that while their work focuses on giving examples of problems of high complexities, such as $\Theta(n^{1/k})$ for $k\in \mathbb{N}$, we settle how the landscape looks for the class of symmetry breaking problems. 

\begin{theorem}[Informal version of \cref{thm:volume_speedup}]
\label{thm:volume_speedup_informal}
If the deterministic or randomized \volume complexity of an LCL is $o(\log^* n)$, it is, in fact, $O(1)$. 
\end{theorem}
We note that together with a result of \cite{brandt_grunau_rozhon2021LLL_in_LCA}, \cref{thm:volume_speedup_informal} implies that the only deterministic \volume complexities for LCLs when identifiers can be from an exponential range are $\Theta(1), \Theta(\log^* n)$ and $\Theta(n)$.
We find it interesting that the \volume complexity landscape is in this regard substantially cleaner than the \local complexity landscape that contains complexities between $\Theta(\log\log^* n)$ and $\Theta(\log^* n)$, while one would, a priori, assume that the usage of a finer measure in the  \volume model will lead to a subtler and more complicated landscape. The reason is simple and sketched next. 

\paragraph*{Our method in a nutshell}
We adapt the Ramsey-theoretical argument of \cite{naorstockmeyer}, which was improved by \cite{chang2017time_hierarchy}, to the \volume model. 
We observe that the heart of their argument is essentially a \volume model argument that shows that algorithms with $o(\log^* n)$ \volume complexity can be made order-invariant and, hence, they can be sped up to complexity $O(1)$. 
The \volume nature of the argument is exactly the reason why the original argument applied to the \local model gives only an $o(\log\log^*n)$ speedup in general graphs and an $o(\sqrt{\log^* n})$ speedup in two-dimensional grids as this is the local complexity that implies $o(\log^* n)$ volume. 
The main conceptual difference between our proof and the original approach \cite{naorstockmeyer,chang2017time_hierarchy} is that \volume model algorithms can perform \emph{adaptive} probes unlike \local algorithms. 

\subsection{Speedup in Oriented Grids}

Our third result completes the \local complexity landscape for LCLs on oriented grids by proving that there are no local complexities between $\omega(1)$ and $o(\log^* n)$.

In an oriented grid, all the edges are oriented in a consistent matter. 
Moreoever, each edge is labeled with a value from $[d]$, indicating the dimension the edge corresponds to. 

Oriented grids may be a bit too special model to consider from the perspective of distributed computing. 
However, in the related fields where ideas from local algorithms can be applied, such the study of factors of iid solutions\cite{holroyd_schramm_wilson2017,holroydliggett2015finitely_dependent_coloring,Spinka,grebik_rozhon2021toasts_and_tails} or constructions in descriptive combinatorics including the famous circle squaring problem \cite{laczkovich,OlegCircle,Circle,JordanCircle,conley_grebik_pikhurko2020divisibility_of_spheres}, oriented grids are often the graph class under consideration. 

\begin{theorem}
\label{thm:main_grids_informal}
[Informal version of \cref{thm:main_grids}]
Let $d$ be a fixed positive constant. Any LCL on a $d$-dimensional oriented grid with local complexity $o(\log^* n)$ has, in fact, local complexity $O(1)$. 
\end{theorem}

The above theorem together with the work of \cite{chang2016exp_separation,chang2017time_hierarchy,brandt_grids} implies the following characterization of possible local complexities on oriented grids. 

\begin{corollary}
Let $d$ be a fixed positive constant. Then the deterministic/randomized complexity of an LCL problem on a $d$-dimensional oriented grid is one of the following:
\begin{enumerate}
    \item $O(1)$, 
    \item $\Theta(\log^* n)$, 
    \item $\Theta(\sqrt[d]{n})$. 
\end{enumerate}
\end{corollary}

Sometimes one considers the case of unoriented grids, that is, just the underlying graph of an oriented grid, without any additional labeling (in the study of randomized processes and descriptive combinatorics, the grid is usually oriented). 
Our result unfortunately does not generalize to this case, as those graphs do not locally induce an implicit order on vertices. 
We conjecture, however, that the local complexity landscape remains the same. 

\begin{conjecture}
On unoriented grids, all LCL problems with complexity $o(\log^* n)$ have local complexity $O(1)$. 
\end{conjecture}

\subsection{Decidability}\label{sec:gendec}
While understanding the complexity landscapes of LCL problems in different models and on different graph classes is certainly interesting in its own right, it also constitutes the first step towards a more ambitious goal---that of \emph{deciding} for a given LCL problem how fast it can be solved, i.e., into which complexity class it falls.
Unfortunately, it is known that this is not possible on general graphs: a result by Naor and Stockmeyer~\cite{naorstockmeyer} states that it is undecidable whether a given LCL can be solved in constant time.
More specifically, the work of Naor and Stockmeyer implies that, on $d$-dimensional oriented grids, it is undecidable whether a given LCL has complexity $O(1)$ or $\Theta(\sqrt[d]{n})$.
For the case of toroidal oriented grids, i.e., grids that wrap around (in the natural sense), it is undecidable whether a given LCL has complexity $O(\log^* n)$ or is a global problem \cite{brandt_grids}. 

However, if we restrict attention to trees, no undecidability results are known, i.e., it is entirely possible that one day a procedure is developed that takes as input an arbitrary LCL and returns its complexity on trees.
While it turned out that decidability questions for LCLs are quite complex already in very simple settings, considerable effort has gone into making progress towards this exciting goal.
On the positive side, it is known that in paths and cycles the only \local complexities are $O(1)$, $\Theta(\log^* n)$, and $\Theta(n)$, and it can be decided in polynomial time into which class a given LCL problem falls, provided that the LCL does not have inputs~\cite{naorstockmeyer,brandt_grids,chang2017time_hierarchy,changauto21}.
When considering general LCLs (i.e., LCLs with inputs) on paths and cycles, it remains decidable which asymptotic complexity a given LCL has, but the question becomes PSPACE-hard~\cite{balliu2019distributed}.
The PSPACE-hardness result extends to the case of LCLs without inputs on constant-degree trees, which was subsequently improved by Chang~\cite{chang2020n1k_speedups} who showed that in this setting the question becomes EXPTIME-hard.

Nevertheless, there are also good news regarding decidability on trees beyond cycles and paths.
For instance, the asymptotic complexity of problems from a natural subclass of LCLs, called \emph{binary labeling problems} can be decided efficiently~\cite{balliu20binary}.
Moreover, in their work providing the complexity classification of LCLs on regular rooted trees~\cite{balliurooted21}, the authors showed that it is decidable into which of the four complexity classes $O(1)$, $O(\log^* n)$, $\Theta(\log n)$, and $n^{\Theta(1)}$ a given LCL without inputs falls.
The decidability is achieved by defining so-called ``certificates'' for $O(\log n)$-round, $O(\log^* n)$-round, and constant-round solvability---for each $T$ from $\{ O(\log n), O(\log^* n), O(1) \}$, an LCL problem is $T$-round solvable if and only if there exists a certificate for $T$-round solvability for the given LCL problem, and the existence of such a certificate is decidable. Unfortunately, the definitions of these certificates rely heavily on the fact that the input tree is rooted (in particular on the availability of a canonical direction in the tree, which is given by the parent-child relation); it is entirely unclear how such an approach could be extended to unrooted trees.

On unrooted trees, Chang and Pettie~\cite{chang2017time_hierarchy} showed that it is decidable whether a given LCL can be solved in logarithmic time or requires polynomial time, and Chang~\cite{chang2020n1k_speedups} showed that in the latter case (i.e., if the given LCL has complexity $\Theta(n^{1/k})$ for some positive integer $k$) the exact exponent is decidable.
However, decidability on (unrooted) trees below $\Theta(\log n)$ is wide open; in fact, even the following simple question is a major open problem in the complexity theory of LCLs.

\begin{question}\label{q:constant}
    Is it decidable whether a given LCL can be solved in constant time on (constant-degree) trees?
\end{question}

Note that constant-time-solvability of LCLs on trees is semidecidable as for each constant $c$ and fixed LCL $\Pi$, there are only constantly many different candidate $c$-round \local algorithms for solving $\Pi$---the difficult direction is to prove (or disprove) semidecidability for the impossibility of constant-time-solvability.
While our proof of \cref{thm:main_informal} does not settle Question~\ref{q:constant}, it may be a first step towards a resolution as it reduces the problem to proving (or disproving) semidecidability of an $\Omega(\log^* n)$-round lower bound for the given LCL.
Note that this avoids consideration of all the ``messy'' complexities shown to exist between $\Theta(\log \log^* n)$ and $\Theta(\log^* n)$ in general (constant-degree) graphs (such as $2^{\Theta(\log^{\alpha} \log^* n)}$ for any positive rational number $\alpha \leq 1$)~\cite{balliu2018new_classes-loglog*-log*}.

\subsection{Organization of the Paper}
In \cref{sec:preliminaries} we formally define the settings that we work with, and prove basic technical results.
In \cref{sec:main_proof} we prove the speedup theorem for trees in the \local model, i.e., \cref{thm:main_informal}. 
In \cref{sec:volume} we prove the speedup theorem in the \volume model, i.e., \cref{thm:volume_speedup_informal}. 
In \cref{sec:grids} we prove the speedup theorem for oriented grids in the \local model, i.e., \cref{thm:main_grids_informal}.

\section{Preliminaries}
\label{sec:preliminaries}

We use classical graph-theoretical notation, e.g. we write $G=(V,E)$ for an unoriented graph. 
A \emph{half-edge} is a pair $h = (v,e)$, where $v \in V$, and $e \in E$ is an edge incident to $v$.
We denote the set of half-edges of $G$ by $H = H(G)$, i.e., $H = \{ (v,e) \mid v \in e, v \in V, e\in E\}$.
Furthermore, for every vertex $v'$, we denote the set of half-edges $(v,e) \in H$ where $v = v'$ by $H[v']$, and for every edge $e'$, we denote the set of half-edges $(v,e) \in H$ where $e = e'$ by $H[e']$.
Often we assume that $G$ additionally carries a labeling of vertices or half-edges.
We use $B_G(u,r)$ to denote the ball of radius $r$ around a node $u$ in $G$ and we call it the \emph{$r$-hop neighborhood of $u$}. 
When talking about half-edges in $B_G(u,r)$, we mean all half-edges $(v,e)$ such that $v \in B_G(u, r)$.
For example, $B_G(u, 0)$ contains all half-edges incident to $u$.

The reader should keep in mind that our setting is graphs of maximum degree bounded by some constant $\Delta$.
This is sometimes explicitly stated (or it is implied by the constraints) but most of the time it is tacitly assumed.
Of special interest to us will be the class of all trees with maximum degree at most $\Delta$, which we denote by $\fT$.
Similarly, we denote the class of all forests with maximum degree at most $\Delta$ by $\fF$.
Moreover, for any positive integer $n$, any set $N$ of positive integers, and any $\fG \in \{ \fF, \fT \}$, we will use $\fG_n$, resp.\ $\fG_N$, to denote the class of members of $\fG$ with $n$ nodes, resp.\ with a number of nodes that is contained in $N$.


\subsection{\local Model and LCL Problems}
\label{subsec:local}

In this section, we define our main model of computation and discuss the problem class considered in this work.
Our main model of computation is the \local model~\cite{linial92}.
Although it is often convenient to think about \local algorithms as message-passing procedures, it will be simpler to work with the following equivalent definition.

\begin{definition}[\local model]
\label{def:local_model}
The input to a problem in the \local model is an $n$-node graph $G$, for some positive integer $n$.
Each node of the graph is considered as a computational entity and equipped with a globally unique identifier, i.e., positive integer from a polynomial range (in the case of deterministic algorithms), or with a private random bit string (in the case of randomized algorithms).
Additionally, depending on the problem considered, other inputs might be stored at the nodes or half-edges.
In a $T$-round algorithm each node $v$ is aware of the number $n$ of nodes of the input graph\footnote{Our results work equally well in the setting where only some upper bound $n'$ on $n$ is given to the nodes; in that case, as usual, the complexities are functions of $n'$ instead of of $n$. We remark that our setting where the nodes are aware of the exact value of $n$ is in fact the more difficult one (for our approach) as evidenced by the issue discussed in Section~\ref{subsec:probseq}.} and of its $T$-hop neighborhood, i.e., of all nodes in distance at most $T$ from $v$, all edges that have at least one endpoint in distance at most $T-1$ from $v$, and all half-edges whose endpoint is in distance at most $T$ from $v$ (as well as all inputs stored therein).
Based on this information, $v$ has to decide on its output whose specification is given by the considered problem.
In other words, a $T$-round algorithm is simply a function (parameterized by $n$) from the the space of all possible (labeled) $T$-hop neighborhoods of a node to the space of outputs.
Which neighborhoods are ``possible'' is determined by the considered problem and the considered graph class.

For technical reasons, we will also assume that each graph comes with a \emph{port numbering}, i.e., each node $v$ has $\deg(v)$ \emph{ports} $1, \dots, \deg(v)$, and each edge incident to $v$ is connected to $v$ via a unique one of those ports.
In other words, the ports at a node provide a total order on the set of incident edges.
It is straightforward to verify that the addition of a port numbering does not change the computational power of the model (asymptotically) as each node can infer a total order on the set of incident edges from the unique identifiers of its neighbors (in the deterministic case) or from the random bits of its neighbors, with arbitrarily large success probability (in the randomized case).
\end{definition}



The class of problems considered in this work are \emph{LCL problems} (or \emph{LCLs}, for short), which were introduced by Naor and Stockmeyer~\cite{naorstockmeyer}.
In their seminal paper, Naor and Stockmeyer provided a definition for LCL problems where input and output labels were assigned to nodes, and remarked that a similar definition can be given for edge-labeling problems.
A modern definition that captures both kinds of LCL problems (and their combinations) assigns labels to \emph{half-edges} (instead of vertices or edges).
Before we can provide this definition, we need to define some required notions.

A \emph{half-edge labeling} of a graph $G$ (with labels from a set $\Sigma$) is a function $f \colon H(G) \to \Sigma$.
A \emph{$\Sigma_{\inn}$-$\Sigma_{\out}$-labeled graph} is a triple $(G, f_{\inn}, f_{\out})$ consisting of a graph $G$ and two half-edge labelings $f_{\inn} \colon H(G) \to \Sigma_{\inn}$ and $f_{\out} \colon H(G) \to \Sigma_{\out}$ of $G$.
We analogously define a \emph{$\Sigma_{\inn}$-labeled graph} by omitting $f_{\out}$.

We can now define an LCL problem as follows.
\begin{definition}[LCL problem]
\label{def:lcl_problem}
An LCL problem $\Pi$ is a quadruple $(\Sigma_{\inn}, \Sigma_{\out}, r, \fP)$ where $\Sigma_{\inn}$ and $\Sigma_{\out}$ are finite sets, $r$ is a positive integer, and $\fP$ is a finite collection of $\Sigma_{\inn}$-$\Sigma_{\out}$-labeled graphs.
A correct solution for an LCL problem $\Pi$ on a $\Sigma_{\inn}$-labeled graph $(G, f_{\inn})$ is given by a half-edge labeling $f_{\out} \colon H(G) \to \Sigma_{\out}$ such that, for every node $v \in V(G)$, the triple $(B_G(v,r), f'_{\inn}, f'_{\out})$ is isomorphic to a member of $\fP$, where $f'_{\inn}$ and $f'_{\out}$ are the restriction of $f_{\inn}$ and $f_{\out}$, respectively, to $B_G(v,r)$.
\end{definition}

Intuitively, the collection $\fP$ provides the constraints of the problem by specifying how a correct output looks \emph{locally}, depending on the respective local input.
From the definition of a correct solution for an LCL problem it follows that members of $\fP$ that have radius $> r$ can be ignored.
Also the finiteness of $\fP$ automatically implies that we are restricting ourselves to graphs of degree at most $\Delta$ for some constant $\Delta$.

The main tool in our proof of \cref{thm:main_informal}, the round elimination technique, applies (directly) only to a subclass of LCL problems: roughly speaking, it is required that the local correctness constraints specified by $\fP$ can be translated into node and edge constraints, i.e., correctness constraints that can be verified by looking at the label configurations on each edge and around each node.
The definition of this subclass of LCL problems is given in the following.
Note that, despite the complicated appearance, the definition is actually quite intuitive: essentially, in order to define the LCL problem, we simply specify a set of label configurations that are allowed on an edge, a set of label configurations that are allowed around a node, and an input-output label relation that specifies for each input label which output label is allowed at the same half-edge.

\begin{definition}[Node-edge-checkable LCL problem]
\label{def:node-edge-checkable_lcl_problem}
    A \emph{node-edge-checkable LCL} $\Pi$ is a quintuple $(\spinn, \spout, \noco, \edco, \gee)$, where $\spinn$ and $\spout$ are finite sets, $\edco$ is a collection of cardinality-$2$ multisets $\{\xB_1, \xB_2\}$ with $\xB_1, \xB_2 \in \spout$, $\noco = (\noco^1, \noco^2, \dots)$ consists of collections $\noco^i$ of cardinality-$i$ multisets $\{\xA_1, \dots, \xA_i\}$ with $\xA_1, \dots, \xA_i \in \spout$, and $\gee \colon \spinn \to 2^{\spout}$ is a function that assigns to each label from $\spinn$ a subset of the labels of $\spout$.
    A correct solution for a node-edge-checkable LCL $\Pi$ on a $\spinn$-labeled graph $(G, f_{\inn})$ is given by a half-edge labeling $f_{\out} \colon H(G) \to \spout$ such that
    \begin{enumerate}
        \item for every node $v \in V(G)$, the multiset consisting of the labels assigned by $f_{\out}$ to the half-edges in $H[v]$ is contained in $\noco^{\deg(v)}$,
        \item for every edge $e \in E(G)$, the multiset consisting of the labels assigned by $f_{\out}$ to the half-edges in $H[e]$ is contained in $\edco$, and
        \item for every half-edge $h \in H(G)$, the label $f_{\out}(h)$ is contained in the label set $\gee(f_{\inn}(h))$.
    \end{enumerate}
    We call $\noco$ the \emph{node constraint} of $\Pi$ and $\edco$ the \emph{edge constraint} of $\Pi$.
    Moreover, we call the elements $\{\xA_1, \dots, \xA_i\}$ of $\noco^i$ \emph{node configurations} and the elements $\{\xB_1, \xB_2\}$ of $\edco$ \emph{edge configurations} (of $\Pi$).
    In a \local algorithm solving a node-edge-checkable problem $\Pi$, each node is supposed to output a label for each incident half-edge such that the induced global half-edge labeling is a correct solution for $\Pi$.
\end{definition}

Even though the round elimination technique can be applied directly only to node-edge-checkable LCL problems, the results we obtain apply to all LCL problems.
The reason for this is that for each LCL problem $\Pi$ there exists a node-edge-checkable LCL problem $\Pi'$ such that the time complexities of $\Pi$ and $\Pi'$ differ only by an additive constant, as we show in Lemma~\ref{lem:LCL_and_nodeedgeLCL_are_same}.
This fact suffices to lift our results for node-edge-checkable LCL problems to general LCL problems: in particular, the existence of an LCL problem with time complexity in $\omega(1)$ and $o(\log^*n)$ would imply the existence of a node-edge-checkable LCL problem with the same complexity constraints, leading to a contradiction.

Before stating and proving Lemma~\ref{lem:LCL_and_nodeedgeLCL_are_same}, we formally define the local failure probability of an algorithm solving a node-edge-checkable LCL, and the complexity of an LCL problem.

\begin{definition}[Local failure probability]
Let $\Pi = (\spinn, \spout, \noco, \edco, \gee)$ be some node-edge-checkable LCL problem.
We say that a half-edge labeling $f_{\out} \colon H(G) \to \Sigma_{\out}$ is \emph{incorrect on some edge $e = \{u, v\}$} of graph $(G, f_{\inn})$ if
\begin{enumerate}
    \item\label{edgecase1} $\{f_{\out}((u,e)), f_{\out}((v,e))\} \notin \edco$, or
    \item\label{edgecase2} $f_{\out}((u,e)) \notin \gee(f_{\inn}((u,e)))$ or $f_{\out}((v,e)) \notin \gee(f_{\inn}((v,e)))$.
\end{enumerate}
Similarly, we say that $f_{\out}$ is \emph{incorrect at some node $v$} if
\begin{enumerate}
    \item\label{nodecase1} $\{f_{\out}(h)\}_{h \in H[v]} \notin \noco^{\deg(v)}$, or
    \item\label{nodecase2} $f_{\out}(h) \notin \gee(f_{\inn}(h))$ for some $h \in H[v]$.
\end{enumerate}
We say that an algorithm $\fA$ \emph{fails on some edge $e$, resp.\ at some node $v$}, if the output produced by $\fA$ is incorrect on $e$, resp.\ at $v$.
Furthermore, we say that a (randomized) algorithm $\fA$ has \emph{local failure probability} $p$ on some graph $G$ if $p$ is the smallest (real) number such that, for each edge $e$ and node $v$ in $G$, the probability that $\fA$ fails on $e$, resp.\ at $v$, is upper bounded by $p$.
Moreover, for each $n$, the local failure probability of $\fA$ on some class of $n$-node graphs is the maximum of the local failure probabilities of $\fA$ on the graphs in the class.
(In contrast, the definition of \emph{(global) failure probability} is as commonly used, i.e., we say that $\fA$ has (global) failure probability $p = p(n)$ if the (worst-case) probability that $\fA$ does not produce a correct solution for $\Pi$ is upper bounded by $p$ and $p$ is minimal under this constraint.)
\end{definition}

The \local complexity of a (node-edge-checkable or common) LCL is simply the minimum complexity of an algorithm $\fA$ that solves it on all graphs.

\begin{definition}[Complexity of an LCL problem]
\label{def:complexity_of_lcl_problem}
The \emph{determinstic (round) complexity} of an LCL $\Pi$ is the function $T \colon \N \rightarrow \N \cup \{ 0 \}$ satisfying that for each $n \in \N$, there exists a deterministic algorithm $\fA_n$ solving $\Pi$ in $T(n)$ rounds on all $n$-node graphs $G$ with each half-edge labeled with a label from $\Sigma_{\inn}$, but no deterministic algorithm solving $\Pi$ in $T(n) - 1$ rounds on this class of graphs.
The \emph{randomized (round) complexity} of an LCL $\Pi$ is defined analogously, where deterministic algorithms are replaced by randomized algorithms with a (global) failure probability of at most $1/n$.

%
\end{definition}

When we talk about the complexity of an LCL on trees, we further restrict the above definition to graphs that are trees (and similarly for other graph classes).

Now we are ready to state and prove the following lemma, which ensures that we can restrict attention to node-edge-checkable LCLs.

\begin{lemma}
\label{lem:LCL_and_nodeedgeLCL_are_same}
    For any LCL problem $\Pi$, there exists a node-edge-checkable LCL problem $\Pi'$ such that (in both the randomized and deterministic \local model) the complexities of $\Pi$ and $\Pi'$ on trees (and on forests) are asymptotically the same.
\end{lemma}
\begin{proof}
    Suppose $\Pi = (\Sigma_\inn, \Sigma_\out, r, \fP)$ is an LCL. We create a node-edge-checkable LCL $\Pi' = (\Sigma_\inn^{\Pi'}, \Sigma_\out^{\Pi'}, \fN_{\Pi'}, \fE_{\Pi'}, g_{\Pi'})$ as follows: 
    \begin{itemize}
        \item $\Sigma_\inn^{\Pi'} = \Sigma_\inn$.
        \item $\Sigma_\out^{\Pi'}$ contains all possible labelings of $r$-hop neighborhoods of a node, each neighborhood has marked a special half-edge, each vertex and each edge has an order on incident half-edges and each half-edge is labeled with a label from $\Sigma_\out$, moreover, the labeling by $\Sigma_\out$ has to be accepted by $\fP$.
        \item $\fN_{\Pi'}$ contains such sets  $S = \{\sigma_1, \dots, \sigma_d\}$ with $\sigma_i \in \Sigma_\out^{\Pi'}$ such that there exists an $r$-hop neighborhood $N$ of a node $u$ of degree $d$ together with each node and each edge having an order on incident half edges and such that half-edges have labels from $\Sigma_\out$ such that we can assign labels from $S$ to half-edges around $u$ in such a way that each $\sigma_i$ assigned to $(u,e_i)$ describes $N$ with the special half-edge of $\sigma_i$ being $e_i$.
        \item $\fE_{\Pi'}$ is defined analogously to $\fN_{\Pi'}$, but we require an existence of a neighborhood of an edge $e$ that is consistent from the perspective of labels from $\Sigma_\out^{\Pi'}$ assigned to $(u,e)$ and $(v,e)$. 
        \item $g_{\Pi'}$ maps each label $\tau \in \Sigma_\inn$ to the set of labels $\sigma \in \Sigma_\out^{\Pi'}$ such that the special half-edge of $\sigma$ is labeled by $\tau$. 
    \end{itemize}
    
    Suppose we have a valid solution of $\Pi$. Then, in $r$ rounds each half-edge can decide on its $\Pi'$-label by encoding its $r$-hop neighborhood, including port numbers of each vertex and each edge that give ordering on its half-edges, into a label from $\Sigma_\out^{\Pi'}$. The constraints $\fN_{\Pi'}, \fE_{\Pi'}, g_{\Pi'}$ will be satisfied. 
    
    On the other hand, suppose we have a valid solution for $\Pi'$. In $0$-rounds, each half-edge $(v,e)$ can label itself with the label on the special half-edge in its $\Pi'$-label $f'_\out((v,e))$. We claim that the $\Pi$-labeling we get this way around a half-edge $(u,e)$ is isomorphic to the $\Sigma_\inn$-labeling described by the label $f_\out((v,e))$, hence the new labeling is a solution to $\Pi'$. To see this, consider running a BFS from $(v,e)$. The node constraints $\fN_{\Pi'}$ and the edge constraints $\fE_{\Pi'}$ are ensuring that the labels from $\Sigma_\out^{\Pi'}$ of visited half-edges are describing compatible neighborhoods, while the function $g$ ensures that the description of $\Sigma_\inn$ labels in the $r$-hop neighborhood of $u$ by the labels from $\Sigma_\out^{\Pi'}$ agrees with the actual $\Sigma_\inn$ labeling of the $r$-hop neighborhood of $u$. As the label $f'_\out((v,e))$ needs to be accepted by $\fP$, we get that $\fP$ accepts the $r$-hop neighborhood of $u$, as needed. 
\end{proof}

It is crucial that the way in which we define the node-edge-checkable LCL problem $\Pi'$ in Lemma~\ref{lem:LCL_and_nodeedgeLCL_are_same} guarantees that the considered (input-labeled) graph class remains the same as for $\Pi$ (and does not turn into a graph class with a promise on the distribution of the input labels, which would be the result of the straightforward approach of defining $\Pi'$ by encoding the input labels contained in a constant-sized ball in $\Pi$ in a single input label in $\Pi'$, and doing the same for output labels).
If this property was not guaranteed, it would be completely unclear (and perhaps impossible) how to extend the round elimination framework of \cite{brandt19automatic_speedup_theorem} to our setting with input labels.

\subsection{Order-Invariant Algorithms}
\label{subsec:other_algorithms}

In this section, we formally define the notion of an order-invariant algorithm and introduce further computational models of interest. We also show that oftentimes order-invariant algorithms can be sped up to improve the round/probe complexity, both in the \local and the \volume model.

\begin{definition}[Order-invariant \local algorithm \cite{naorstockmeyer}]
\label{def:order-invariant}
A deterministic $T(n)$-round \local algorithm $\fA$ is called order-invariant if the following holds: Consider two assignments of distinct identifiers to nodes in $B_G(v,T(n))$ denoted by $\mu$ and $\mu'$. Assume that for all $u,w \in B_G(v,T(n))$ it holds that $\mu(u) > \mu(w)$ if and only if $\mu'(u) > \mu'(w)$. Then, the output of $\mathcal{A}$ on the set of half-edges $H[v]$ will be the same in both cases.
\end{definition}

Next, we define the notions necessary for \cref{sec:volume}. We start by defining the \volume model \cite{Rosenbaum2020}. We define the \volume model in a more mathematically rigorous way compared to \cite{Rosenbaum2020}, as this will help us later with the proofs. Before defining the \volume model, we start with one more definition.

\begin{definition}
 For an arbitrary $S \subseteq \mathbb{N}$, we define
 \[Tuples_S = \{(id,deg,in) \colon id \in S, deg \in [\Delta], in \colon [deg] \mapsto \Sigma_{in}\}\] and for $i \in \mathbb{N}$, we define

\[Tuples_{i,S} = \{(t_1,t_2,\ldots,t_i) \colon \text{for every $j \in [i]$, $t_j \in Tuple_S$}\}.\]
 
For a given $i \in \mathbb{N}$ and $\ell \in [2]$, let
\[t^{(\ell)} = ((id^\ell_1,deg_1,in_1), (id^\ell_2,deg_2,in_2), \ldots, (id^\ell_i,deg_i,in_i)) \in Tuples_{i,\mathbb{N}}\] be arbitrary.
We say that the tuples $t^{(1)}$ and $t^{(2)}$ are almost identical if for every $j_1,j_2 \in [i]$, $id^1_{j_1} < id^2_{j_2}$ implies $id^2_{j_1} < id^2_{j_2}$, $id^1_{j_1} > id^1_{j_2}$ implies $id^2_{j_1} > id^2_{j_2}$ and $id^1_{j_1} = id^1_{j_2}$ implies $id^2_{j_1} = id^2_{j_2}$.
\end{definition}

A tuple in $Tuples_S$ can encode the local information of a node $v$, including its ID, its degree and the input assigned to each of its incident half-edges, i.e., $in(k)$ is the input assigned to the $k$-th half edge. 
We denote with $t_v$ the tuple that encodes the local information of $v$.
A tuple in $Tuples_{i,S}$ can be used to encode all the information (modulo the number of nodes of the input graph) that a node knows about the input graph after having performed $i-1$ probes.
The notion of almost identical tuples will be helpful for defining the notion of order-invariance for the \volume model.

\begin{definition}[\volume model]

Let $\Pi = (\Sigma_{in}, \Sigma_{out}, r, \mathcal{P})$ be an LCL problem. 
A deterministic \volume model algorithm $\mathcal{A}$ for $\Pi$ with a probe complexity of $T(n)$ can be seen as a set of computable functions $f_{n,i}$ for every $n\in \mathbb{N}$ and $i \in [T(n) + 1]$ with $f_{n,i} \colon Tuples_{i,\mathbb{N}} \mapsto [i] \times [\Delta]$ for $i \in [T(n)]$ encoding the $i$-th adaptive probe and $f_{n,T(n)+1} \colon Tuples_{T(n) + 1,\mathbb{N}} \mapsto \Sigma_{out}^{[\Delta]}$, where $\Sigma_{out}^{[\Delta]}$ refers to the set of functions mapping each value in $[\Delta]$ to a value in $\Sigma_{out}$, encoding the output that $\mathcal{A}$ assigns to the incident half edges of the queried node.
Next, we have to define what it means for $\mathcal{A}$ to be a valid algorithm.
To that end, let $(G,f_{in})$ be an arbitrary $\Sigma_{in}$-labeled graph on $n$ nodes with each node in $G$ being equipped with an identifier and a port assignment (To simplify the technical definition, we assume that $G$ does not contain any isolated node).

The algorithm $\mathcal{A}$ defines a half-edge labeling  $f_{out,\mathcal{A},(G,f_{in})} \colon H(G) \mapsto \Sigma_{out}$ as follows:
Let $(v,e) \in H(G)$ be an arbitrary half-edge of $G$. We define $t^{(0)} = (t_v)$ and for $i \in \{1,2,\ldots,T(n)\}$, we obtain $t^{(i)}$ from $t^{(i-1)} = (t_{v_0},t_{v_1},\ldots,t_{v_i})$ as follows. Let $(j,p) = f_{n,i}(t^{(i-1)})$ and $v_{i+1}$ the node in $G$ such that $\{v_j,v_{i+1}\}$ is the $p$-th edge incident to $v_j$ (we assume that the degree of $v_j$ is at least $p$). Then, $t^{(i)} = (t_{v_0},t_{v_1},\ldots,t_{v_{i+1}})$. Finally, $f_{n,T(n) + 1}(t_{v_0},t_{v_1},\ldots,t_{v_{T(n)}})$ defines a function $g$ mapping each value in $[\Delta]$ to a value in $\Sigma_{out}$. Let $e$ be the $p$-th edge incident to $v$. We define $f_{out,\mathcal{A},(G,f_{in})}((v,e)) = g(p)$.

We say that $\mathcal{A}$ solves $\Pi$ if $f_{out,\mathcal{A},(G,f_{in})}$ is a valid output labeling for every $\Sigma_{in}$-labeled graph $(G,f_{in})$ with each node in $G$ having a unique ID from a polynomial range, is equipped with a port labeling, and $G$ does not contain an isolated node.

\label{def:volume}

\end{definition}

The notion of an order-invariant algorithm naturally extends to algorithms in the \volume model.
\begin{definition}(Order-invariant \volume algorithm)
We say that $\mathcal{A}$ is order invariant if for every $n \in \mathbb{N}$ and $i \in [T(n)+1]$, $f_{n,i}(t) = f_{n,i}(t')$ for every $t,t' \in Tuples_{i,S}$ with $t$ and $t'$ being almost identical.
\end{definition}

We will use the following basic speed-up result for order-invariant algorithms in both the \local and the \volume model.
\begin{theorem}[Speed-up of order-invariant algorithms (cf. \cite{chang2016exp_separation})]
\label{thm:speedup_order_invariant}
Let $\fA$ be an order-invariant algorithm solving a problem $\Pi$ in $f(n) = o(\log n)$ rounds of the \local model or with $f(n) = o(n)$ probes of the \volume model. 
Then, there is an order-invariant algorithm solving $\Pi$ in $O(1)$ rounds of the \local model or, equivalently, $O(1)$ probes in the \volume model. 
\end{theorem}
\begin{proof}
The result for the \local model is proven in \cite{chang2016exp_separation}. The proof for the \volume model follows in the exact same manner. We provide it here for completeness.
Let $n_0$ be a fixed constant such that $\Delta^{r+1} \cdot (T(n_0) + 1) \leq n_0/\Delta$.
We now define a \volume algorithm $\fA'$ with probe complexity $T'(n) = \min(n,T(n_0)) = O(1)$ as follows.
For $n \in \mathbb{N}$ and $i \in [T'(n) + 1]$, we define $f^{\fA'}_{n,i} = f^{\fA}_{\min(n,n_0),i}$. 
It remains to show that $\fA'$ indeed solves $\Pi$.
For the sake of contradiction, assume that this is not the case. 
This implies the existence of a $\Sigma_{in}$-labeled graph $(G,f_{in})$ on $n \geq n_0$ nodes (with IDs from a polynomial range, port assignments, and no isolated nodes) such that $\fA'$ "fails" on $(G,f_{in})$.
Put differently, there exists a node $v$ such that $\fA'$ produces a mistake in the $r$-hop neighborhood of $v$.
The $r$-hop neighborhood of $v$ consists of at most $\Delta^{r+1}$ vertices. To answer a given query, $\fA'$ "sees" at most $T(n_0) + 1$ nodes. Hence, to compute the output of all the nodes in the $r$-hop neighborhood of $v$, $\fA'$ "sees" at most  $\Delta^{r+1}(T(n_0) + 1) \leq \frac{n_0}{\Delta}$ many nodes. We denote the set consisting of those nodes as $V^{visible}$.
Now, let $(G',f'_{in})$ be a $\Sigma_{in}$-labeled graph on $n'$ nodes (with IDs from a polynomial range, port assignments, and no isolated nodes) such that every $u \in V^{visible}$ is also contained in $G'$, with its degree being the same in both graphs, as well as the input assigned to each of its incident half edges. The assigned ID can be different, however, the relative orders of the IDs assigned to nodes in $V^{visible}$ in $G$ and $G'$ are the same. As $\Delta^{r+1}(T(n_0) + 1) \leq \frac{n_0}{\Delta}$, such a $(G',f'_{in})$ exists.  
As $\fA$ is order invariant, so is $\fA'$.
Moreover, $f^{\fA'}_{n,i} = f^{\fA'}_{n',i}$ for any $i$. 
Hence, it follows that $\fA'$ assigns the same output to all the half-edges in the $r$-hop neighborhood of $v$ in $G$ and $G'$. Therefore, $\fA'$ also fails on the graph $G'$. From the way we defined $\fA'$, this directly implies that $\fA$ also fails on $G$, a contradiction with the assumption that $\fA$ is a correct algorithm. This finishes the proof.
\end{proof}

\paragraph{\lca model}
We now briefly discuss the related \lca model.
A deterministic local computation algorithm (\lca) is similar to a deterministic \volume algorithm, with two small differences. First, an \lca can perform so-called far probes and second, it can assume that each node in the $n$-node input graph has a unique ID in the set $\{1,2,\ldots,n\}$. However, far probes are not of any help in the complexity regime we consider. 

\begin{theorem}[cf. \cite{goos_nonlocal_probes_16}, Theorem 1]
Any LCL problem that can be solved by an \lca with probe complexity $t(n)$ can also be solved by an \lca with probe complexity $O(t(n^{\log(n)}))$ that does not perform any far probes, provided $t(n) = o(\sqrt{log(n)})$.
\end{theorem}

In particular, an \lca  with probe complexity $o(\log^*n)$ implies an \lca for the same problem that does not perform any far probes and has a probe complexity of $o(\log^* n)$. Hence, we can focus on showing that any \lca that does not perform any far probes with a probe complexity of $o(\log^* n)$ implies an \lca for the same problem with a probe complexity of $O(1)$. However, an \lca that does not perform any far probes is the same as a \volume algorithm, with the only difference being that the \lca algorithm only has to produce a valid output if each node in the $n$-node input graph has a unique ID from the set $\{1,2,\ldots,n\}$. However, one can show with a very simple argument that a \volume algorithm (\lca  without far probes) with probe complexity $T(n) = o(\log^*n)$ that assumes unique IDs from the set $\{1,2,\ldots,n\}$ implies a volume algorithm  with probe complexity $T'(n) = T(n^k) = o(\log^*n)$ that only assumes IDs from the set $\{1,2,\ldots,n^k\}$ for an arbitrary constant $k$. Hence, from the discussion above, a speed-up result from $o(\log^*n)$ to $O(1)$ in the \volume model directly implies the same speed-up in the \lca model.

\section{The \local Model Gap on Trees}
\label{sec:main_proof}

In this section we prove the $\omega(1) - o(\log^* n)$ gap for LCLs on trees in the \local model.
We do so by proving Theorem~\ref{thm:treemain} (which is the slightly more formal version of Theorem~\ref{thm:main_informal}) by explicitly designing, for any given (node-edge-checkable) LCL problem $\Pi$ with complexity $o(\log^* n)$, a constant-round algorithm.
As explained in Section~\ref{subsec:introtree}, a very rough outline of our approach is to generate from $\Pi$ a sequence of node-edge-checkable LCL problems of decreasing randomized local complexities (where we allow the local failure probability to grow along the problems in the sequence), find a problem in the sequence that can be solved in $0$ rounds with a reasonably low local failure probability, show that there exists a $0$-round deterministic algorithm for that problem, and turn this algorithm into a constant-round algorithm for $\Pi$ by going back up the sequence of problems and arguing that the deterministic complexities increase slowly along the sequence in this direction.
While the round elimination framework~\cite{brandt19automatic_speedup_theorem,balliu2019LB} provides a blueprint how to generate a suitable sequence, it unfortunately only does so for LCLs on regular trees without inputs.
We provide an extension of the framework that also works for LCLs on irregular trees (or forests) with inputs.

We will start in Section~\ref{subsec:probseq} by extending the definition of the round elimination problem sequence to the setting with inputs (and taking care of a technical issue).
In Section~\ref{subsec:goingdown}, we will carefully bound the evolution of failure probabilities along a sequence of algorithms with decreasing runtimes that solve the problems in the defined problem sequence.
Section~\ref{subsec:goingup} takes care of the reverse step, i.e., showing that the deterministic complexities of the problems in the problem sequence do not increase fast when traversed towards $\Pi$.
Finally, in Section~\ref{subsec:alltog}, we will put everything together and prove Theorem~\ref{thm:treemain}.

\subsection{The Problem Sequence}\label{subsec:probseq}

Similarly to the approach in~\cite{brandt19automatic_speedup_theorem}, we define, for any node-edge-checkable LCL problem $\Pi$, two node-edge-checkable problems $\re(\Pi)$ and $\rere(\Pi)$.
The problems in the aforementioned sequence are then obtained by iteratively applying $\rere(\re(\cdot))$, starting with $\Pi$.


\begin{definition}[$\fR(\Pi)$]
\label{def:repi}
Let $\Pi = (\spinn, \spout, \noco, \edco, \gee)$ be a node-edge-checkable LCL problem.
We define a new node-edge-checkable LCL problem $\re(\Pi) = (\rspinn, \rspout, \rnoco, \redco, \rgee)$ by specifying the five components.
We start by setting $\rspinn := \spinn$ and $\rspout := 2^{\spout}$, i.e., the input label set of $\re(\Pi)$ is simply the input label set of $\Pi$, and the output label set of $\re(\Pi)$ is the power set of the output label set of $\Pi$.
Next, we define $\rgee$ by setting $\rgee(\ell) := 2^{\gee(\ell)}$ for any label $\ell \in \spinn$, i.e., intuitively speaking, in problem $\re(\Pi)$ an input label $\ell$ on some half-edge requires that the output label on the same half-edge is a subset of the set of output labels that were allowed in $\Pi$ on a half-edge with input label $\ell$.

We define the edge constraint $\redco$ of $\re(\Pi)$ as the set of all cardinality-$2$ multisets $\{ \xB_1, \xB_2 \}$ such that $\xB_1, \xB_2 \in \rspout$ and, for all $\xb_1 \in \xB_1$, $\xb_2 \in \xB_2$, we have $\{\xb_1, \xb_2\} \in \edco$.
Finally, we define the node constraint $\rnoco$ of $\re(\Pi)$ as follows.
For each integer $i \geq 1$, define $\rnoco^i$ as the set of all cardinality-$i$ multisets $\{ \xA_1, \dots, \xA_i \}$ such that $\xA_1, \dots, \xA_i \in \rspout$ and there exists some selection $(\xa_1, \dots, \xa_i)$ of labels from $\xA_1 \times \dots \times \xA_i$ such that $\{\xa_1, \dots, \xa_i\} \in \noco^i$.

\end{definition}

Note that our definition of $\re(\Pi)$ differs slightly from the usual definition of $\re(\Pi)$ as given in, e.g., \cite{balliu2020ruling} (beyond the obvious differences due to the fact that we consider LCL problems with input labels): our definition does not remove so-called ``non-maximal'' configurations.
Removing such configurations can be beneficial when trying to determine the complexity of specific problems, but is not required (or helpful) in our setting, where we want to argue about the complexity of many problems at once.

\begin{definition}[$\rere(\Pi)$]
\label{def:rerepi}
The problem $\rere(\Pi)$ differs from $\re(\Pi)$ only in the node and edge constraints; for the remaining three parameters we set $\orrspinn := \rspinn$, $\orrspout := \rspout$, and $\orrgee := \rgee$.

We define the node constraint $\orrnoco$ of $\rere(\Pi)$ as follows.
For each integer $i \geq 1$, define $\orrnoco^i$ as the set of all cardinality-$i$ multisets $\{ \xA_1, \dots, \xA_i \}$ such that $\xA_1, \dots, \xA_i \in \orrspout$ and, for all $(\xa_1, \dots, \xa_i) \in \xA_1 \times \dots \times \xA_i$, we have $\{\xa_1, \dots, \xa_i\} \in \noco^i$.
Moreover, we define the edge constraint $\orredco$ of $\rere(\Pi)$ as the set of all cardinality-$2$ multisets $\{ \xB_1, \xB_2 \}$ such that $\xB_1, \xB_2 \in \orrspout$ and there exists some selection $(\xb_1, \xb_2)$ of labels from $\xB_1 \times \xB_2$ such that $\{\xb_1, \xb_2\} \in \edco$.
\end{definition}

Note that, although the function $\rere(\cdot)$ can take any arbitrary node-edge-checkable LCL problem as argument, we will use as arguments only problems that are of the form $\re(\Pi)$ for some node-edge-checkable LCL problem $\Pi$.

Recall that $\fT$, resp.\ $\fF$, denotes the class of all trees, resp.\ forests, of maximum degree at most $\Delta$. 
Before turning to the analysis of the evolution of the aforementioned failure probabilities, there is a technical issue we have to discuss.
A crucial argument in said analysis is, roughly speaking, that if you consider some (sufficiently small) neighborhood of a node (or edge), and a set of extensions of this neighborhood via different edges leaving the neighborhood\footnote{An extension via some leaving edge is simply a possibility of how the graph could continue for the next hop beyond the respectively chosen leaving edge that is consistent with (some graph in) the considered graph class.}, then there must also be a graph in the considered graph class that (simultaneously) contains all extensions of the set (together with the neighborhood).
Here the term ``considered graph class'' describes the input graph class restricted to the members that are consistent with the knowledge of the nodes about the number $n$ of nodes, i.e., in particular if all nodes are aware of the exact value of $n$ (as they are in our definition of the \local model), then the considered graph class contains only $n$-node graphs. 
Now, if the input graph class is $\fT$, it follows that the aforementioned crucial argument does not hold in general: if all extensions in the considered set ``conclude'' the tree (i.e., do not have leaving edges except those connecting them to the considered initial neighborhood) and the combined number of nodes in the initial neighborhood and the considered extensions does not happen to be precisely $n$, then there is no $n$-node tree that contains the neighborhood together with all considered extensions.

We solve this issue by proving our main theorem first for the class $\fF$ of forests (which do not have the aforementioned issue as the number of nodes\footnote{We remark that the issue does not occur in the variant of the \local model in which nodes are only aware of some upper bound on the number of nodes, but we believe that it is important to ensure that the correctness of (our) results does not depend on such minor details in the model specification.} in some maximal connected component is not known to the nodes) in Theorem~\ref{thm:REmain} and then lifting it to $\fT$ in Theorem~\ref{thm:treemain} by showing that, for node-edge-checkable LCL problems, a complexity of $o(\log^* n)$ on trees implies a complexity of $o(\log^* n)$ on forests.
Lemma~\ref{lem:treetoforest} provides this relation between trees and forests; in Sections~\ref{subsec:goingdown} and~\ref{subsec:goingup} we will then exclusively work with forests.
Note that the complexity of any LCL problem $\Pi$ on $\fF$ is trivially at least as large as its complexity on $\fT$ as $\fT$ is a subclass of $\fF$.

\begin{lemma}\label{lem:treetoforest}
    Let $\Pi$ be some node-edge-checkable LCL problem that has deterministic, resp.\ randomized, complexity $o(\log^* n)$ on $\fT$.
    Then the deterministic, resp.\ randomized, complexity of $\Pi$ on $\fF$ is in $o(\log^* n)$.
\end{lemma}
\begin{proof}
    Let $\fA$ be a (deterministic or randomized) algorithm solving $\Pi$ on $\fT$ in $T(n) \in o(\log^* n)$ rounds (and observe that the existence of $\fA$ (even if it is randomized) implies that a correct global solution exists).
    We design a new algorithm $\fA'$ solving $\Pi$ on $\fF$ in $o(\log^* n)$ rounds.
    Algorithm $\fA'$ proceeds as follows on any $n$-node input forest $G' \in \fF$, where, for any node $u$, we denote the (maximal) connected component containing $u$ by $C_u$ and the number of nodes in $C_u$ by $|C_u|$.
    
    First, each node collects its ($2T(n^2)+2$)-hop neighborhood in $G'$.
    Then, based on the collected information, each node $u$ determines whether there exists a node $v$ in $C_u$ such that the ($T(n^2)+1$)-hop neighborhood of $v$ contains all of $C_u$.
    
    If such a node $v$ exists, then each node in $C_u$ is aware of the whole component $C_u$ and can simply choose the same solution for $\Pi$ on $C_u$ (by mapping component $C_u$ (including unique identifiers or random bits) in some arbitrary, but fixed, deterministic fashion to some correct solution), and then output the part of the solution it is responsible for.
    Note that this requires all nodes in $C_u$ to be distinguishable from each other (so that each node knows which part of the solution on $C_u$ it is responsible for); for deterministic algorithms this is guaranteed by the unique identifiers, for randomized algorithms it is guaranteed with probability at least $1 - 1/n^2$ by having each node interpret its first $\lceil 4 \log n \rceil$ random bits as an identifier (which guarantees uniqueness of the created identifiers with probability at least $1 - 1/n^2$).
    
    If no such node $v$ exists, then $u$ simply executes $\fA$ with input parameter\footnote{Recall that each node receives as input a parameter representing the number of nodes. Note that nothing prevents us from executing an algorithm using an input parameter that does not represent the correct number of nodes.} $n^2$ (and each node in $C_u$ will do likewise).
    In this case, due to the fact that no ($T(n^2)+1$)-hop node neighborhood fully contains $C_u$, it holds for each node $v$ in $C_u$ that the ($T(n^2)+1$)-hop neighborhood of $v$ in $G'$ is isomorphic to the ($T(n^2)+1$)-hop neighborhood of some node $w$ in some $n^2$-node tree $G \in \fT$.
    Hence, if $\fA'$ fails on some node $v$ in $C_u$ or on some edge incident to $v$, then $\fA$ fails on some node $w$ in some $n^2$-node tree, or on some edge incident to $w$.
    Since the failure probability of $\fA$ on $n^2$-node trees is at most $1/n^2$, it holds for any node $w$ in any $n^2$-node tree that the probability that $\fA$ fails on $w$ or an edge incident to $w$ is at most $1/n^2$.
    It follows for each node $v$ in $C_u$ that the probability that $\fA'$ fails on $v$ or an edge incident to $v$ is at most $1/n^2$.
    
    Now, a union bound over all components $C_u$ of the first kind (``such a node $v$ exists'') and all nodes in components $C_u$ of the second kind (``no such node $v$ exists'') yields that $\fA'$ fails with probability at most $1/n$.
    Note that if $\fA$ is deterministic, then all of the above failure probabilities are $0$, and $\fA'$ is deterministic as well.
    
    For the runtime of $\fA'$, observe that in either of the two considered cases, the initial collection of $u$'s ($2T(n^2)+2$)-hop neighborhood suffices to compute $u$'s output in $\fA'$.
    Since $T(n) \in o(\log^* n)$ implies $2T(n^2)+2 \in o(\log^* n)$, it follows that the runtime of $\fA'$, and therefore also the complexity of $\Pi$ on $\fF$, is in $o(\log^* n)$.
\end{proof}

\subsection{From Harder to Easier Problems}\label{subsec:goingdown}

Recall that, for any set $N$ of positive integers, $\fF_N$ denotes the class of forests with a number of nodes that is contained in $N$.
The goal of this section is to prove the following theorem.

\begin{theorem}\label{thm:fullspeedup}
    Let $\Pi$ be a node-edge-checkable LCL problem and $\fA$ a randomized algorithm solving $\Pi$ on $\fF$ with runtime $T(n)$ and local failure probability at most $p \leq 1$.\footnote{Note that (bounds on) local failure probabilities such as $p$ also (possibly) depend on the number $n$ of nodes; however, for better readability we will omit this dependency.}
    Let $N$ be the set of all positive integers $n$ satisfying $T(n) + 2 \leq \log_{\Delta} n$.
    Then, there exists a randomized algorithm $\fA'$ solving $\rere(\re(\Pi))$ on $\fF_N$ with runtime $\max\{0, T(n) - 1\}$ and local failure probability at most $Sp^{1/(3\Delta + 3)}$, where
    \begin{equation*}
        S = (10\Delta(\lvert \spinn \rvert + \max\{\lvert \spout \rvert, \lvert \rspout \rvert\}))^{4\Delta^{T(n)+1}}.
    \end{equation*}
\end{theorem}
In other words, we want to show, roughly speaking, that we can solve $\rere(\re(\Pi))$ at least one round faster than $\Pi$ if we allow the indicated increase in the local failure probability of the randomized algorithm solving the problem.

For the proof of Theorem~\ref{thm:fullspeedup}, we will make use of an approach that is an extension of the approach known for the setting of LCLs on regular trees without inputs (see, e.g., \cite{balliu2019LB,balliu2020ruling,balliuseek21}).
More specifically , we will explicitly define an algorithm $\fA'$ (depending on $\fA$) that satisfies the properties stated in Theorem~\ref{thm:fullspeedup}.
While in the regular setting without inputs, the possibilities how the input graph could continue beyond the view of a node $v$ differ only in the random bits of the nodes beyond $v$'s view, the irregularity and inputs in our setting require us to first perform a simulation step in the definition of $\fA'$ that simulates all possible topologies and inputs that a node could encounter (not too far) beyond its view before considering the randomness contained in each such extension defined by the combination of topology and inputs.
As we show, the resulting more complex definition of $\fA'$ still allows us to give an upper bound on the increase of the local failure probability from $\fA$ to $\fA'$ that suffices for our purposes.

It should be noted that, due to the fact that we consider a large class of LCL problems at once (and not a single fixed LCL problem, in which case certain simplification techniques might be applicable), the increase in the number of output labels from $\Pi$ to $\rere(\re(\Pi))$ is doubly exponential.
Since the bound on the local failure probability of $\fA'$ depends (superlinearly) on the number of output labels of $\Pi$, and, ultimately, we want to apply Theorem~\ref{thm:fullspeedup} iteratively (starting with $\Pi$), we cannot apply Theorem~\ref{thm:fullspeedup} more than $\Theta(\log^* n)$ times before the local failure probability grows too large.
This provides a (different) explanation why we cannot extend the $\omega(1) - o(\log^* n)$ gap further with our approach (which also follows from the fact that there are problems with complexity $\Theta(\log^* n)$).

Note that we will define $\fA'$ for all forests from $\fF$, but will prove the guarantee on the local failure probability of $\fA'$ only for the forests in $\fF_N$.
Also recall that, for any node $u$ in a graph $G$, we denote the $r$-hop neighborhood of $u$ in $G$ by $B_G(u,r)$.
By abuse of notation, we will use $B_G(u,r)$ both for all the information contained in the $r$-hop neighborhood of $u$ (i.e., the topology, inputs, and random bits) and for the respective subgraph of $G$ (including input information, but no random bits).

\paragraph*{Deriving $\fA'$}
Let $\fA$ be a randomized algorithm for some node-edge-checkable LCL problem $\Pi$ with runtime $T = T(n)$.
If $T = 0$, we can simply let $\fA'$ simulate $\fA$, and then, for each half-edge $h$, transform the intermediate output $\ell \in \spout$ returned by $\fA$ on $h$ into the final output $\{\{\ell\}\} \in \rrspout$ on $h$.
By construction, $\fA'$ fails on some edge, resp.\ node, if and only if $\fA$ fails on the same edge, resp.\ node; since $p \leq Sp^{1/(3\Delta + 3)}$ for the $S$ specified in Theorem~\ref{thm:fullspeedup}, it follows that $\fA'$ satisfies the properties required in Theorem~\ref{thm:fullspeedup}.
Hence, we will assume in the following that $T \geq 1$.

In order to derive $\fA'$ from $\fA$, we first derive an ``intermediate'' algorithm $\fA_{1/2}$ for $\re(\Pi)$ from $\fA$, and then we derive $\fA'$ from $\fA_{1/2}$.
As we will show later, $\fA_{1/2}$ solves $\re(\Pi)$ (on $\fF_N$) with a moderately increased local failure probability (and, intuitively speaking, very slightly reduced runtime) compared to $\fA$; a similar moderate increase in local failure probability (and slight decrease in runtime) is incurred when going from $\fA_{1/2}$ to $\fA'$.

Deviating from the usual convention that nodes are the entities performing the computation, we will assume for $\fA_{1/2}$ that the edges of the input forest perform the computation.
This is not in contradiction with the definition of the \local model as $\fA_{1/2}$ is only a construct defined for the design of $\fA'$; the actual computation in $\fA'$ is performed by the nodes.
Algorithm $\fA_{1/2}$ proceeds as follows.

Each edge $e = \{u,v\}$ in the input forest $G \in \fF$ first collects all information (i.e., the topology, inputs, and random bits) contained in the union $B_G(e,T - 1/2) := B_G(u,T - 1) \cup B_G(v,T - 1)$ of the ($T - 1$)-hop neighborhoods of $u$ and $v$.
Then, $e$ determines the output label $\ell'$ it outputs on half-edge $(u,e)$ as follows, depending on some parameter $0 < K \leq 1$ that we will choose later.
Label $\ell'$ is simply the set of all labels $\ell$ such that there exists an input forest $G' \in \fF$ (including input labels) and an edge $e'$ in $G'$ such that $B_{G'}(e', T - 1/2) \cong B_G(e, T - 1/2)$ and the probability that the node $u'$ corresponding to $u$ in the isomorphism outputs $\ell$ on $(u',e')$ according to $\fA$ is at least $K$, conditioned on
the assumption that the random bits in $B_{G'}(e', T - 1/2)$ are the same as in $B_G(e, T - 1/2)$.
Here, the isomorphism is w.r.t.\ the topology and the input labels.
In other words, in $\fA_{1/2}$, edge $e$ outputs on $(u,e)$ the set of all labels $\ell$ for which the probability that, in $\fA$, node $u$ outputs $\ell$ on $(u,e)$, conditioned on the random bits that $e$ has collected, is at least $K$ for at least one possible extension of (the topology and input labels of) the graph beyond the $(T - 1/2)$-hop view of $e$.
Edge $e$ computes the output label on half-edge $(v,e)$ analogously.
This concludes the description of $\fA_{1/2}$; in the following we derive $\fA'$ from $\fA_{1/2}$ in a fashion dual to how we derived $\fA_{1/2}$ from $\fA$.
	
In $\fA'$, each node $u$ first collects all information contained in $B_G(u, T - 1)$.
Then, for each incident edge $e$, node $u$ determines the output label $\ell''$ it outputs on half-edge $(u,e)$ as follows, depending on some parameter $0 < L \leq 1$ that we will choose later.
Label $\ell''$ is simply the set of all labels $\ell'$ such that there exists an input forest $G'' \in \fF$ and a node $u''$ in $G''$ such that $B_{G''}(u'', T - 1) \cong B_G(u, T - 1)$ and the probability that the edge $e''$ corresponding to $e$ in the isomorphism outputs $\ell'$ on $(u'',e'')$ according to $\fA_{1/2}$ is at least $L$, conditioned on
the assumption that the random bits in $B_{G''}(u'', T - 1)$ are the same as in $B_G(u, T - 1)$.
In other words, in $\fA'$, node $u$ outputs on $(u,e)$ the set of all labels $\ell'$ for which the probability that, in $\fA_{1/2}$, edge $e$ outputs $\ell'$ on $(u,e)$, conditioned on the random bits that $u$ has collected, is at least $L$ for at least one possible extension of (the topology and input labels of) the graph beyond the $(T - 1)$-hop view of $u$.
This concludes the description of $\fA'$.

In the following, for all forests in $\fF_N$, we bound the local failure probability of $\fA_{1/2}$, depending on (the bound on) the local failure probability of $\fA$.
We start by proving two helper lemmas.
For any edge $e$, resp.\ node $u$, let $p^*_e$, resp.\ $p^*_u$, denote the probability that $\fA_{1/2}$ fails at edge $e$, resp.\ node $u$.
Moreover, for graphs $G$, $G'$, let $f_{\inn}$, resp.\ $f'_{\inn}$, be the functions that, for each half-edge in $G$, resp.\ $G'$, return the input label of the half-edge.
Finally, we will use $\fA(h)$ and $\fA_{1/2}(h)$ to denote the labels that $\fA$ and $\fA_{1/2}$, respectively, output on some half-edge $h$.

\begin{lemma}\label{lem:edgebound}
    Let $e = \{u, v\}$ be an arbitrary edge in an arbitrary forest $G \in \fF_N$.
    It holds that $p^*_e \leq ps/(K^2)$, where $s = (3\lvert \spinn \rvert)^{2\Delta^{T+1}}$.
\end{lemma}

\begin{proof}
	Consider an arbitrary assignment of random bits in $B_G(e, T - 1/2)$ for which $\fA_{1/2}$ fails on $e$, i.e., for which
	\begin{enumerate}
	    \item\label{badcase1} the cardinality-$2$ multiset of labels that $\fA_{1/2}$ outputs on $(u,e)$ and $(v,e)$ is not contained in $\redco$ \\ , i.e., $\{\fA_{1/2}((u,e)), \fA_{1/2}((v,e))\} \notin \redco$, or
	    \item\label{badcase2} we have $\fA_{1/2}((u,e)) \notin \rgee(f'_{\inn}((u,e)))$ or $\fA_{1/2}((v,e)) \notin \rgee(f'_{\inn}((v,e)))$.
	\end{enumerate}
	First, consider the case that Condition~\ref{badcase1} is satisfied.
	Then, by the definition of $\re(\Pi)$, there are two labels $\ell_u \in \fA_{1/2}((u,e))$, $\ell_v \in \fA_{1/2}((v,e))$ such that $\{\ell_u, \ell_v\} \notin \edco$.
	Moreover, by the definition of $\fA_{1/2}$, there is a forest $G' \in \fF_N$ containing\footnote{For better readability, we refrain from using the mathematically precise term ``isomorphism'' in the following, and instead identify isomorphic objects with each other, e.g., we consider $B_G(e, T - 1/2)$ to be a subgraph of $G'$ if it is isomorphic to some subgraph of $G'$.} $B_G(e, T - 1/2)$ such that $|V(G')| = |V(G)|$ and, conditioned on the already fixed random bits in $B_{G'}(e, T - 1/2) = B_G(e, T - 1/2)$, the probability that $\fA$, when executed on $G'$, returns $\ell_u$ on $(u,e)$ and $\ell_v$ on $(v,e)$ is at least $K^2$.
	Note that we use here that the input graph is a forest: in order to be able to simply multiply the two probabilities of $\ell_u$ being returned on $(u,e)$ and $\ell_v$ being returned on $(v,e)$ (which are both lower bounded by $K$), we require that those two probabilities are independent (which is guaranteed if the input graph is a forest, as then $B_{G'}(u, T) \setminus B_{G'}(e, T - 1/2)$ and $B_{G'}(v, T) \setminus B_{G'}(e, T - 1/2)$ are disjoint).
	Note further that we use that $G \in \fF_N$ (which implies that the number of nodes in $B_G(e, T - 1/2)$ is sufficiently small compared to $|V(G)|$) to guarantee the property $|V(G')| = |V(G)|$, and that this property is needed because the lemma statement relating the probabilities $p^*_e$ and $p$ (which technically speaking are functions of the number $n$ of nodes) is supposed to hold for any fixed $n$.
	Finally, note that this is a place where it is crucial that we consider forests, not trees, as on trees is might be the case that no graph $G'$ as described exists: if the two extensions beyond $B_G(e, T - 1/2)$ that are responsible for the containment of $\ell_u$ in $\fA_{1/2}((u,e))$ and of $\ell_v$ in $\fA_{1/2}((v,e))$ both have no edges leaving the respective extension except those connecting them to $B_G(e, T - 1/2)$, and the total number of nodes in the union of $B_G(e, T - 1/2)$ and those two extensions does not happen to be precisely $n$, then those two extensions cannot appear simultaneously if we assume the input graph to be an $n$-node tree.
	
	Now consider the case that Condition~\ref{badcase2} is satisfied.
	Then, by the definition of $\re(\Pi)$, there is some label $\ell_u \in \fA_{1/2}((u,e))$ satisfying $\ell_u \notin \gee(f'_{\inn}((u,e)))$ or some label $\ell_v \in \fA_{1/2}((v,e))$ satisfying $\ell_v \notin \gee(f'_{\inn}((v,e)))$. With an analogous argumentation to the one used in the previous case, we obtain that there is a forest $G' \in \fF_N$ containing $B_G(e, T - 1/2)$ such that $|V(G')| = |V(G)|$ and the probability that $\fA$, when executed on $G'$, fails on $e$ is at least $K$.
	
	Since $0 < K \leq 1$, we can conclude that in either case the probability (conditioned on the already fixed random bits in $B_{G'}(e, T - 1/2)$) that, in $G'$, $\fA$ fails on $e$ is at least $K^2$.
	Observe that the output of $\fA$ on the two half-edges belonging to $e$ depends only on $B_{G'}(e, T + 1/2) = B_{G'}(u,T) \cup B_{G'}(u,T)$.
	Given the fixed topology and input in $B_{G'}(e, T - 1/2)$, there are at most $s := (3\lvert \spinn \rvert)^{2\Delta^{T+1}}$ different possibilities for the topology and input in $B_{G'}(e, T + 1/2)$: there are at most $2\Delta^T$ nodes in $B_{G'}(e, T + 1/2) \setminus B_{G'}(e, T - 1/2)$, and for each of the $\Delta$ possible ports of such a node, there are at most $3$ possibilities regarding topology (being connected to a node in $B_{G'}(e, T - 1/2)$, being connected to a node outside of $B_{G'}(e, T + 1/2)$, or non-existing due to the degree of the node being too small) and at most $\lvert \spinn \rvert$ possibilities regarding the input label on the half-edge corresponding to that port.
	Let $\fB$ denote the set of different balls of the form $B_{G'}(e, T + 1/2)$ (for some $G'$) that contain $B_G(e, T - 1/2)$.
	As shown above, $\lvert \fB \rvert \leq s$.
	
	Now, forget the fixing of the random bits in $B_G(e, T - 1/2)$.
	By the above discussion, it follows that there is some ball $B \in \fB$ (and therefore also some forest $G'$) containing $B_G(e, T - 1/2)$ such that the probability that $\fA$, when executed on $B$ (or $G'$), fails on $e$ is at least $(p^*_e/s) \cdot K^2$.
	Since this probability is upper bounded by $p$, we obtain $p^*_e \leq ps/(K^2)$, as desired.
\end{proof}

\begin{lemma}\label{lem:nodebound}
    Let $u$ be an arbitrary node in an arbitrary forest $G \in \fF_N$.
    It holds that $p^*_u \leq p + \lvert \spout \rvert \Delta K + \frac{ps\Delta}{K}$, where $s = (3\lvert \spinn \rvert)^{2\Delta^{T+1}}$.
\end{lemma}

\begin{proof}
	Let $e_1, \dots, e_{\deg(u)}$ denote the edges incident to $u$.
	Moreover, let $p^{(1)}_u$ denote the probability that \\ $\{\fA_{1/2}((u, e_i))\}_{1 \leq i \leq \deg(u)} \notin \rnoco^{\deg(u)}$ and $p^{(2)}_u$ the probability that there exists some $1 \leq i \leq \deg(u)$ satisfying $\fA_{1/2}((u, e_i))\} \notin \rgee(f_{\inn}((u,e_i)))$.
	Since those two conditions together cover all cases in which $\fA_{1/2}$ fails at $u$, we have $p^*_u \leq p^{(1)}_u + p^{(2)}_u$.
	We start by bounding $p^{(1)}_u$.
	
	Observe that correctness at $u$, for both $\fA$ and $\fA_{1/2}$, depends only on $B_G(u, T)$.
	Given the topology and input in $B_G(u, T)$ (which is already fixed since we are considering some fixed graph $G$), we call, for each label $\ell \in \spout$ and each $1 \leq i \leq \deg(u)$, an assigment of random bits in $B_G(u, T)$ \emph{bad for the pair $(\ell,i)$} if $\fA((u, e_i)) = \ell$ and $\ell \notin \fA_{1/2}((u, e_i))$ (under this assignment).
	Observe that the definition of $\fA_{1/2}$ ensures, for each fixed assignment of random bits in $B_G(e_i, T-1/2)$, that if $\ell \notin \fA_{1/2}((u, e_i))$ under this assignment\footnote{Note that $\fA_{1/2}((u, e_i))$ is uniquely determined after fixing the random bits in $B_G(e_i, T-1/2)$.}, then the probability that $\fA((u, e_i)) = \ell$ (conditioned on the fixed assignment in $B_G(e_i, T-1/2)$) is smaller than $K$.
	Hence, (prior to fixing any random bits) for each pair $(\ell,i) \in \spout \times \{1, \dots, \deg(u)\}$, the probability that an assignment of random bits in $B_G(u,T)$ is bad for $(\ell,i)$ is smaller than $K$.
	It follows by a union bound that the probability that an assignment of random bits is bad for \emph{some} pair $(\ell,i) \in \spout \times \{1, \dots, \deg(u)\}$ is upper bounded by $\lvert \spout \rvert \cdot \Delta \cdot K$.
	
	Now, consider an arbitrary assignment of random bits in $B_G(u,T)$ such that $\{\fA_{1/2}((u, e_i))\}_{1 \leq i \leq \deg(u)} \notin \rnoco^{\deg(u)}$ under this assignment.
	We argue that then $\fA$ fails at $u$ or there is some pair $(\ell,i) \in \spout \times \{1, \dots, \deg(u)\}$ such that the assignment is bad for $(\ell,i)$: indeed, if there is no such pair and $\fA$ does not fail at $u$, then, for each $1 \leq i \leq \deg(u)$, we have $\fA((u, e_i)) \in \fA_{1/2}((u, e_i))$, which (combined again with the correctness of $\fA$ at $u$) would imply, by the definition of $\rnoco^{\deg(u)}$, that $\{\fA_{1/2}((u, e_i))\}_{1 \leq i \leq \deg(u)} \in \rnoco^{\deg(u)}$, contradicting the stated property of the considered assignment.
	We conclude that $p^{(1)}_u \leq p + \lvert \spout \rvert \cdot \Delta \cdot K$.
	
	Next, we bound $p^{(2)}_u$.
	By a union bound, it follows from the definition of $p^{(2)}_u$ that there is some $1 \leq i \leq \deg(u)$ such that the probability that $\fA_{1/2}((u, e_i)) \notin \rgee(f_{\inn}((u,e_i)))$ is at least $p^{(2)}_u/\deg(u) \geq p^{(2)}_u/\Delta$.
	
	Consider an arbitrary assignment of random bits in $B_G(e_i, T - 1/2)$ such that $\fA_{1/2}((u, e_i)) \notin \rgee(f_{\inn}((u,e_i)))$.
	Then, by the definition of $\re(\Pi)$, there is some label $\ell \in \fA_{1/2}((u, e_i))$ satisfying $\ell \notin \gee(f_{\inn}((u, e_i)))$.
	Moreover, by the definition of $\fA_{1/2}$, there is some ball $B_{G'}(u, T)$ (in some forest $G' \in \fF_N$ satisfying $|V(G')| = |V(G)|$) containing $B_G(e_i, T - 1/2)$ such that, conditioned on the already fixed random bits in $B_{G'}(e_i, T - 1/2) = B_G(e_i, T - 1/2)$, the probability that $\fA$, when executed on $B_{G'}(u, T)$ (or $G'$), returns $\ell$ on $(u,e_i)$ is at least $K$.
	Note that since $\ell \notin \gee(f_{\inn}((u, e_i)))$, returning $\ell$ on $(u, e_i)$ implies that $\fA$ fails at $u$.
	
	As already established in the proof of Lemma~\ref{lem:edgebound}, there are at most $s := (3\lvert \spinn \rvert)^{2\Delta^{T+1}}$ different possibilities for the topology and input of a ball $B_{G'}(e_i, T + 1/2)$ containing $B_G(e_i, T - 1/2)$, which implies the same bound on the number of different balls $B_{G'}(u, T)$ containing $B_G(e_i, T - 1/2)$.
	Analogously to before (and forgetting the fixing of the random bits in $B_G(e_i, T - 1/2)$), we obtain that there is some ball $B_{G'}(u, T)$ (in some forest $G'$ satisfying $|V(G')| = |V(G)|$) containing $B_G(e_i, T - 1/2)$ such that the probability that $\fA$, when executed on $B_{G'}(u, T)$ (or $G'$), fails at $u$ is at least $((p^{(2)}_u/\Delta)/s) \cdot K$.
	Since this probability is upper bounded by $p$, we obtain $p^{(2)}_u \leq ps\Delta/K$, which implies
	\begin{equation}
	    p^*_u \leq p + \lvert \spout \rvert \Delta K + \frac{ps\Delta}{K} \enspace,
	\end{equation}
	as desired.
\end{proof}

By using Lemmas~\ref{lem:edgebound} and \ref{lem:nodebound} and choosing $K$ suitably, we now give an upper bound on the local failure probability of $\fA_{1/2}$ on $\fF_N$.

\begin{lemma}\label{lem:firsthalfstep}
    Algorithm $\fA_{1/2}$ has local failure probability at most $2\Delta(s + \lvert \spout \rvert)p^{1/3}$ on $\fF_N$, where $s = (3\lvert \spinn \rvert)^{2\Delta^{T+1}}$.
\end{lemma}

\begin{proof}
    Choose the parameter in the definition of $\fA_{1/2}$ as $K := p^{1/3}$.
    Then, by Lemma~\ref{lem:nodebound}, for each node $u$ in the input forest $G$, we obtain
    \begin{equation*}
        p^*_u \leq p + \lvert \spout \rvert \Delta p^{1/3} + s\Delta p^{2/3} \leq 2\Delta(s + \lvert \spout \rvert)p^{1/3} \enspace,
    \end{equation*}
    and, by Lemma~\ref{lem:edgebound}, for each edge $e$ in $G$, we obtain
    \begin{equation*}
        p^*_e \leq s p^{1/3} \leq 2\Delta(s + \lvert \spout \rvert)p^{1/3} \enspace.
    \end{equation*}
    The lemma statement follows by the definition of local failure probability.
\end{proof}

Next, we prove an analogous statement to Lemma~\ref{lem:firsthalfstep} relating the local failure probabilities of $\fA_{1/2}$ and $\fA'$.
For any edge $e$, resp.\ node $u$, let $p'_e$, resp.\ $p'_u$, denote the probability that $\fA'$ fails at edge $e$, resp.\ node $u$.

\begin{lemma}\label{lem:secondhalfstep}
    If $\fA_{1/2}$ has local failure probability at most $p^* \leq 1$ on $\fF_N$, then $\fA'$ has local failure probability at most $3(s + \lvert \rspout \rvert)(p^*)^{1/(\Delta + 1)}$ on $\fF_N$, where $s = (3\lvert \spinn \rvert)^{2\Delta^{T+1}}$.
\end{lemma}

\begin{proof}
    We start by obtaining analogous statements to Lemmas~\ref{lem:edgebound} and \ref{lem:nodebound}.
    By simply exchanging the roles of nodes and edges\footnote{Note that, as usual for round elimination, all proofs also directly extend to hypergraphs. Hence, exchanging the roles of nodes and edges is very natural: a node $v$ of degree $\deg(v)$ simply becomes a hyperedge containing $\deg(v)$ endpoints, while an edge becomes a node of degree $2$.} and reducing the radii of the considered balls by $1/2$ (as well as using parameter $L$ instead of $K$, and using $\rspout$ instead of $\spout$), we directly obtain analogous proofs resulting in the following two statements, (assuming the stated upper bound $p^*$ on the local failure probability of $\fA_{1/2}$):
    \begin{enumerate}
        \item For any node $u$ in the input forest $G$, we have $p'_u \leq p^*s/(L^{\Delta})$.
        \item For any edge $e$ in $G$, we have $p'_e \leq p^* + \lvert \rspout \rvert \cdot 2L + p^*\cdot 2s/L$.
    \end{enumerate}
    Note that, compared to Lemmas~\ref{lem:edgebound} and \ref{lem:nodebound}, each $2$ has been replaced by $\Delta$ and vice versa, simply because the roles of edges (that end in $2$ nodes) and nodes (that ``end'' in at most $\Delta$ edges) are reversed in the analogous proofs.
    Also observe that, technically, the expression for $s$ that we obtain in the new proofs is $(3\lvert \spinn \rvert)^{\Delta \cdot \Delta^{(T-1)+1}}$; however, since this is upper bounded by the original expression for $s$, the obtained statements also hold for the original $s = (3\lvert \spinn \rvert)^{2\Delta^{T+1}}$ (which we will continue to use).
    
    Now, analogously to the choice of $K = p^{1/(2+1)}$ in the proof of Lemma~\ref{lem:firsthalfstep}, set $L := (p^*)^{1/(\Delta + 1)}$.
    We obtain, for each edge $e$ in $G$,
    \begin{equation*}
        p'_e \leq p^* + \lvert \rspout \rvert \cdot 2(p^*)^{1/(\Delta + 1)} + 2s (p^*)^{\Delta/(\Delta + 1)} \leq 3(s + \lvert \rspout \rvert)(p^*)^{1/(\Delta + 1)} \enspace,
    \end{equation*}
    and, for each node $u$ in $G$,
    \begin{equation*}
        p'_u \leq s (p^*)^{1/(\Delta + 1)} \leq 3(s + \lvert \rspout \rvert)(p^*)^{1/(\Delta + 1)} \enspace.
    \end{equation*}
    Again, the lemma statement follows by the definition of local failure probability.
\end{proof}

Combining Lemmas~\ref{lem:firsthalfstep} and \ref{lem:secondhalfstep}, we are finally ready to prove Theorem~\ref{thm:fullspeedup}.

\begin{proof}[Proof of Theorem~\ref{thm:fullspeedup}]
    Let $S$ be as specified in the theorem, and let $\fA'$ be as derived before.
    As already argued during the definition of $\fA'$, if the runtime $T = T(n)$ of $\fA$ is $0$, then $\fA'$ satisfies the stated properties regarding runtime and local failure probability.
    Hence, assume in the following that $T \geq 1$.
    
    From the definition of $\fA'$, it follows directly that the runtime of $\fA'$ is $T - 1$.
    It remains to show that $\fA'$ has local failure probability at most $Sp^{1/(3\Delta + 3)}$ on $\fF_N$.
    We start by showing that $\fA'$ has local failure probability at most $S'p^{1/(3\Delta + 3)}$ on $\fF_N$, where
    \begin{equation*}
        S' := 3 \cdot ((3\lvert \spinn \rvert)^{2\Delta^{T+1}} + \lvert \rspout \rvert) \cdot (2\Delta((3\lvert \spinn \rvert)^{2\Delta^{T+1}} + \lvert \spout \rvert))^{1/(\Delta + 1)} \enspace.
    \end{equation*}
    Indeed, if $p$ satisfies $2\Delta((3\lvert \spinn \rvert)^{2\Delta^{T+1}} + \lvert \spout \rvert)p^{1/3} \leq 1$, then this is a direct consequence of Lemmas~\ref{lem:firsthalfstep} and \ref{lem:secondhalfstep}.
    If $p$ does not satisfy the given condition, then the straightforward combination of Lemmas~\ref{lem:firsthalfstep} and \ref{lem:secondhalfstep} does not work, as the condition $p^* \leq 1$ in Lemma~\ref{lem:secondhalfstep} is not satisfied.
    However, in this case, i.e., if $2\Delta((3\lvert \spinn \rvert)^{2\Delta^{T+1}} + \lvert \spout \rvert)p^{1/3} > 1$, we obtain $S'p^{1/(3\Delta + 3)} > 1$, which trivially implies that $\fA'$ has local failure probability at most $S'p^{1/(3\Delta + 3)}$.
    
    Now, the theorem statement follows from the fact that $S' \leq S$, which in turn follows from
    \begin{align*}
        S' &= 3 \cdot ((3\lvert \spinn \rvert)^{2\Delta^{T+1}} + \lvert \rspout \rvert) \cdot (2\Delta((3\lvert \spinn \rvert)^{2\Delta^{T+1}} + \lvert \spout \rvert))^{1/(\Delta + 1)}\\
        &\leq (9(\lvert \spinn \rvert + \lvert \rspout \rvert))^{2\Delta^{T+1}} \cdot (6\Delta (\lvert \spinn \rvert + \lvert \spout \rvert))^{2\Delta^{T+1}}\\
        &\leq (10\Delta(\lvert \spinn \rvert + \max\{\lvert \spout \rvert, \lvert \rspout \rvert\}))^{2\Delta^{T+1}} \cdot (10\Delta (\lvert \spinn \rvert + \max\{\lvert \spout \rvert, \lvert \rspout \rvert\}))^{2\Delta^{T+1}}\\
        &= (10\Delta(\lvert \spinn \rvert + \max\{\lvert \spout \rvert, \lvert \rspout \rvert\}))^{4\Delta^{T+1}}\\
        &= S \enspace.
    \end{align*}
\end{proof}

\subsection{From Easier to Harder Problems}\label{subsec:goingup}

In this section, we prove the following lemma, which, in a sense, provides a counterpart to Theorem~\ref{thm:fullspeedup}: it states that the time needed to solve some problem $\Pi$ is not much larger than the time needed to solve $\rere(\re(\Pi))$.
It is a simple extension of one direction of \cite[Theorem 4.1]{brandt19automatic_speedup_theorem} to the case of graphs with inputs.
We note that the lemma also holds for randomized algorithms (with essentially the same proof), but we will only need the deterministic version.

\begin{lemma}\label{lem:up}
    Let $\Pi$ be a node-edge-checkable LCL problem and let $\fA$ be a deterministic algorithm solving $\rere(\re(\Pi))$ in $T(n)$ rounds, for some function $T$ (on some arbitrary class of graphs).
    Then there exists a deterministic algorithm $\fA'$ solving $\Pi$ in $T(n) + 1$ rounds.
\end{lemma}
\begin{proof}
    We will define $\fA'$ as follows, depending on $\fA$.\footnote{Note that the $\fA$ and $\fA'$ considered here are not the same as the ones considered in Section~\ref{subsec:goingdown}.}
    
    For each half-edge $h$, let $\fA(h)$ denote the output label that $\fA$ outputs at $h$.
    In $\fA'$, each node $v$ starts by computing $\fA(h)$ for each half-edge $h = (w,e)$ such that $w$ is a neighbor of or identical to $v$ and $e$ is an edge incident to $v$.
    As, in $\fA$, each neighbor of $v$ computes its output in $T(n)$ rounds, $v$ can compute all mentioned $\fA(h)$ in $T(n) + 1$ rounds, by simulating $\fA$.
    Node $v$ will decide on its output solely based on the information collected so far, which implies that the runtime of $\fA'$ is $T(n) + 1$.
    
    For choosing its output, $v$ proceeds in two steps (computed without gathering any further information).
    In the first step, for each incident edge $e = \{ v, w \}$, node $v$ chooses, in some deterministic fashion, a label $L_{(v,e)} \in \fA((v,e))$ and a label $L_{(w,e)} \in \fA((w,e))$ such that $\{ L_{(v,e)}, L_{(w,e)} \} \in \redco$.
    Such a pair of labels exists, by the definition of $\rredco$ (and the fact that $\fA$ correctly solves $\rere(\re(\Pi))$).
    Moreover, the definition of $\rrnoco^{\deg(v)}$ implies that the multiset $\{ L_{(v,e')} \}_{e' \ni v}$ is contained in $\rnoco^{\deg(v)}$, and the definition of $\rrgee$ implies that $L_{(v,e)} \in \rgee(i_{(v,e)})$ and $L_{(w,e)} \in \rgee(i_{(w,e)})$ where $i_{(v,e)}, i_{(w,e)} \in \spinn$ denote the input labels assigned to $(v,e)$ and $(w,e)$, respectively.
    Note also that since the labels $L_{(v,e)}$ and $L_{(w,e)}$ are chosen in a deterministic fashion, depending only on $\fA((v,e))$ and $\fA((w,e))$, node $v$ and node $w$ will compute the same label $L_{(v,e)}$ and the same label $L_{(w,e)}$.
    We conclude that labeling each half-edge $h$ with label $L_{h}$ yields a correct solution for $\re(\Pi)$.
    
    In the second step, $v$ computes the final output for each half-edge incident to $v$, in a fashion analogous to the first step.
    More precisely, for each incident edge $e$, node $v$ chooses a final output label $\ell_{(v,e)} \in L_{(v,e)}$ such that the multiset $\{ \ell_{(v,e')} \}_{e' \ni v}$ is contained in $\noco^{\deg(v)}$.
    Such labels $\ell_{(v,e)}$ exist, by the definition of $\rnoco^{\deg(v)}$ (and the fact that the labeling assigning $L_h$ to each half-edge $h$ correctly solves $\re(\Pi)$).
    Moreover, the definition of $\redco$ implies that, for any edge $e = \{ v, w \}$, the multiset $\{ \ell_{(v,e)}, \ell_{(w,e)} \}$ is contained in $\edco$, and the definition of $\rgee$ implies that, for any half-edge $(v,e)$, we have $\ell_{(v,e)} \in \gee(i_{(v,e)})$.
    
    We conclude that labeling each half-edge $h$ with label $\ell_{h}$ yields a correct solution for $\Pi$.
    It follows that $\fA'$ solves $\Pi$ in $T(n) + 1$ rounds, as desired.
\end{proof}

\subsection{Proving the Gap}\label{subsec:alltog}

In this section, we will finally prove our main result that on forests or trees, any LCL problem with complexity $o(\log^* n)$ can be solved in constant time.
We will first take care of the case of forests.

\begin{theorem}\label{thm:REmain}
    Let $\Pi$ be an arbitrary LCL problem that has (deterministic or randomized) complexity $o(\log^* n)$ on $\fF$.
    Then $\Pi$ can be solved in constant time on $\fF$ (both deterministically and randomized).
\end{theorem}
\begin{proof}
    Observe that, by Lemma~\ref{lem:LCL_and_nodeedgeLCL_are_same}, it suffices to prove the theorem for node-edge-checkable LCL problems.
    Hence, assume in the following that $\Pi$ is node-edge-checkable (and has deterministic or randomized complexity $o(\log^* n)$ on $\fF$).

    Observe further that if each node independently chooses an identifier (of length $O(\log n)$ bits) uniformly at random from a set of size $n^3$, then the probability that there are two nodes that choose the same identifier is at most $1/n$, by a union bound over all pairs of nodes.
    Hence, a deterministic algorithm solving $\Pi$ in $o(\log^* n)$ rounds can be transformed into a randomized algorithm with the same runtime and a failure probability of at most $1/n$ by having each node create its own identifier from its random bits before executing the deterministic algorithm.
    Thus, we can assume that the $o(\log^* n)$-round algorithm for $\Pi$ guaranteed in the theorem is randomized and fails with probability at most $1/n$.
    Let $\fA$ denote this algorithm, and observe that the bound of $1/n$ on the failure probability of $\fA$ implies that also the local failure probability of $\fA$ is bounded by $1/n$.
    
    Let $T(n)$ denote the runtime of $\fA$ on the class of $n$-node forests from $\fF$.
    Let $n_0$ be a sufficiently large positive integer; in particular, we require that
    \begin{equation}\label{eq:nodebound}
        T(n_0) + 2 \leq \log_{\Delta} n_0,
    \end{equation}
    \begin{equation}\label{eq:logstarbound}
        2T(n_0) + 5 \leq \log^* n_0,
    \end{equation}
    and
    \begin{equation}\label{eq:complexbound}
        \left((S^*)^2 \cdot (\log n_0)^{2\Delta}\right)^{(3\Delta + 3)^{T(n_0)}} < n_0
    \end{equation}
    where $S^* := (10\Delta(\lvert \spinn \rvert + \log n_0))^{4\Delta^{T(n_0)+1}}$.
    Since $T(n) \in o(\log^* n)$ (and $\Delta$ and $\spinn$ are fixed constants), such an $n_0$ exists.
    Moreover, for simplicity, define $f(\cdot) := \rere(\re(\cdot))$, and recall that $\fF_{n_0}$ denotes the class of all forests from $\fF$ with $n_0$ nodes.
    Repeatedly applying $f(\cdot)$, starting with $\Pi$, yields a sequence of problems $\Pi, f(\Pi), f^2(\Pi), \dots$.
    Our first goal is to show that there is a $0$-round algorithm solving $f^{T(n_0)}(\Pi)$ on $\fF_{n_0}$ with small local failure probability, by applying Theorem~\ref{thm:fullspeedup} repeatedly.
    
    Recall that for any integer $i \geq 0$, we have $\sout^{\re(f^i(\Pi))} = 2^{\sout^{f^i(\Pi)}}$ and $\sout^{\rere(\re(f^i(\Pi)))} = 2^{\sout^{\re(f^i(\Pi))}}$, by definition.
    This implies that for any $0 \leq i \leq T(n_0)$, we have
    \[
        \max\{\lvert \sout^{f^i(\Pi)} \rvert, \lvert \sout^{\re(f^i(\Pi))} \rvert\} \leq 2^{2^{\iddots^{2^{\spout}}}}
    \]
    where the power tower is of height $2T(n_0) + 3$.
    By (\ref{eq:logstarbound}), we obtain that
    \begin{equation}\label{eq:labelbound}
        \max\{\lvert \sout^{f^i(\Pi)} \rvert, \lvert \sout^{\re(f^i(\Pi))} \rvert\} \leq \log n_0 \text{ for all } 0 \leq i \leq T(n_0).
    \end{equation}
    Recall that $S^* = (10\Delta(\lvert \spinn \rvert + \log n_0))^{4\Delta^{T(n_0)+1}}$ and note that $\sinn^{f^{T(n_0)}(\Pi)} = \sinn^{f^{T(n_0) - 1}(\Pi)} = \ldots = \spinn$.
    By the above discussion, we conclude that when applying Theorem~\ref{thm:fullspeedup} to problem $f^i(\Pi)$ for some $0 \leq i \leq T(n_0)$, then the parameter $S$ in the obtained upper bound $Sp^{1/(3\Delta + 3)}$ (in Theorem~\ref{thm:fullspeedup}) is upper bounded by $S^*$.
    Hence, by applying Theorem~\ref{thm:fullspeedup} $T(n_0)$ times, each time using the upper bound $S^*p^{1/(3\Delta + 3)}$ (instead of $Sp^{1/(3\Delta + 3)}$ with the respective $S$ from Theorem~\ref{thm:fullspeedup}), we obtain that there exists a randomized algorithm $\fA^*$ solving $f^{T(n_0)}(\Pi)$ on $\fF_{n_0}$ in $0$ rounds with local failure probability $p^*$ at most
    \begin{align*}
        &S^* \cdot (S^*)^{1/(3\Delta + 3)} \cdot (S^*)^{1/(3\Delta + 3)^2} \cdot \ldots \cdot (S^*)^{1/(3\Delta + 3)^{T(n_0) - 1}} \cdot 1/(n_0)^{1/(3\Delta + 3)^{T(n_0)}}\\
        &= (S^*)^{\sum_{i = 0}^{T(n_0) - 1} (1/(3\Delta + 3)^i)} \cdot 1/(n_0)^{1/(3\Delta + 3)^{T(n_0)}}\\
        &\leq (S^*)^2 \cdot 1/(n_0)^{1/(3\Delta + 3)^{T(n_0)}}\\
        &< 1/(\log n_0)^{2\Delta}\\
        &\leq 1/\left(\sout^{f^{T(n_0)}(\Pi)}\right)^{2\Delta},
    \end{align*}
    where the second-to-last inequality follows from (\ref{eq:complexbound}), and the last inequality from (\ref{eq:labelbound}).
    Note that (\ref{eq:nodebound}) guarantees the applicability of Theorem~\ref{thm:fullspeedup} to the graph class $\fF_{n_0}$ for our purposes.
    Moreover, we can assume that $\fA^*$ outputs only labels from $\sout^{f^{T(n_0)}(\Pi)}$ (even if it fails), as otherwise we can turn $\fA^*$ into such an algorithm (with the same runtime and a smaller or equally small local failure probability) by simply replacing each label that $\fA^*$ outputs at some node and is not contained in $\sout^{f^{T(n_0)}(\Pi)}$ by some arbitrary label from $\sout^{f^{T(n_0)}(\Pi)}$.
    
    Our next step is to show that the obtained bound on the local failure probability $p^*$ of $\fA^*$ implies that there exists a \emph{deterministic} algorithm solving $f^{T(n_0)}(\Pi)$ on $\fF_{n_0}$ in $0$ rounds.
    To this end, let $\fI$ be the set of all tuples $I = (i_1, \dots, i_k)$ consisting of $k \in \{ 1, \dots, \Delta \}$ labels from $\spinn$, and $\fO$ the set of all tuples $O = (o_1, \dots, o_k)$ consisting of $k \in \{ 1, \dots, \Delta \}$ labels from $\sout^{f^{T(n_0)}(\Pi)}$.
    We say that $I = (i_1, \dots, i_k)$, resp.\ $O = (o_1, \dots, o_k)$, is the \emph{input tuple}, resp.\ \emph{output tuple}, of some node $v$ if $\deg(v) = k$ and, for each $1 \leq j \leq k$, we have that $i_j$ is the input label assigned to, resp.\ the output label returned by $\fA^*$ at, the half-edge incident to $v$ corresponding to port $j$ at $v$.
    
    Since the runtime of $\fA^*$ is $0$ rounds, any node $v$ chooses its output tuple solely based on its random bit string and its input tuple.
    We now define a function $\fA_{\det} \colon \fI \rightarrow \fO$ (that will constitute the desired deterministic algorithm) as follows.
    Consider an arbitrary tuple $I = (i_1, \dots, i_k) \in \fI$ and let $v$ be a node with input label $I$.
    Since $\fA^*$ outputs only labels from $\sout^{f^{T(n_0)}(\Pi)}$ and $k \leq \Delta$, there are at most $\left(\sout^{f^{T(n_0)}(\Pi)}\right)^{\Delta}$ different possibilities for the output tuple of $v$. 
    Hence, there exists a tuple $O \in \fO$ that $v$ outputs with probability at least $1/\left(\sout^{f^{T(n_0)}(\Pi)}\right)^{\Delta}$ (when executing $\fA^*$).
    Fix such an $O$ arbitrarily, and set $\fA_{\det}(I) = O$.
    This concludes the description of $\fA_{\det}$.
    Note that the definition of $\fA_{\det}$ is independent of the choice of $v$ (under the given restrictions) as the only relevant parameters are the input tuple at $v$ (which is fixed) and the random bit string at $v$ (which comes from the same distribution for each node).
    
    We claim that for any two (not necessarily distinct) configurations $I, I' \in \fI$, and any two (not necessarily distinct) labels $o \in \fA_{\det}(I)$ and $o' \in \fA_{\det}(I')$, we have $\{o, o'\} \in \fE_{f^{T(n_0)}(\Pi)}$.
    For a contradiction, assume that there are two configurations $I = (i_1, \dots, i_k), I' = (i'_1, \dots, i'_{k'})$ from $\fI$ and two labels $o \in \fA_{\det}(I)$ and $o' \in \fA_{\det}(I')$ satisfying $\{o, o'\} \notin \fE_{f^{T(n_0)}(\Pi)}$.
    Let $O = (o_1, \dots, o_k)$ and $O' = (o'_1, \dots, o'_{k'})$ be the tuples that $I$ and $I'$, respectively, are mapped to by $\fA_{\det}$, i.e., $\fA_{det}(I) = O$ and $\fA_{\det}(I') = O'$.
    Let $j \in \{ 1, \dots, k \}$ and $j' \in \{ 1, \dots, k' \}$ be two ports/indices such that $o_j = o$ and $o'_{j'} = o'$.
    Consider a forest (with $n_0$ nodes) containing two adjacent nodes $v, v'$ such that $I$ is the input tuple of $v$, $I'$ is the input tuple of $v'$, and the edge $e := \{ v, v' \}$ corresponds to port $j$ at $v$ and to port $j'$ at $v'$.
    By the definition of $\fA_{\det}$, we know that, when executing $\fA^*$, the probability that $v$ outputs $o$ on half-edge $(v,e)$ and the probability that $v'$ outputs $o'$ on half-edge $(v',e)$ are each at least $1/\left(\sout^{f^{T(n_0)}(\Pi)}\right)^{\Delta}$.
    Since these two events are independent (as one depends on the random bit string of $v$ and the other on the random bit string of $v'$), it follows that the probability that the output on edge $e$ is $\{ o, o' \}$ is at least $1/\left(\sout^{f^{T(n_0)}(\Pi)}\right)^{2\Delta}$.
    Now, $\{ o, o' \} \notin \fE_{f^{T(n_0)}(\Pi)}$ yields a contradiction to the fact the local failure probability of $\fA^*$ is strictly smaller than $1/\left(\sout^{f^{T(n_0)}(\Pi)}\right)^{2\Delta}$, proving the claim.
    
    Observe that for any configuration $I = (i_1, \dots, i_k)$ from $\fI$, we also have that
    \begin{enumerate}
        \item $\fA_{\det}(I)$ (or, more precisely, the unordered underlying multiset) is contained in $\fN_{f^{T(n_0)}(\Pi)}^k$, and
        \item for each port $1 \leq j \leq k$, the entry from $\fA_{\det}(I)$ corresponding to port $j$ is contained in $g_{f^{T(n_0)}(\Pi)}(i_j)$,
    \end{enumerate}
    as otherwise for each node $v$ with input tuple $I$, algorithm $\fA^*$ would fail at $v$ with probability at least $1/\left(\sout^{f^{T(n_0)}(\Pi)}\right)^{\Delta}$, yielding again a contradiction to the aforementioned upper bound on the local failure probability of $\fA^*$.
    By the above discussion, we conclude that the output returned by $\fA_{\det}$ is correct at each node and each edge, and therefore $\fA_{\det}$ constitutes a deterministic $0$-round algorithm solving $f^{T(n_0)}(\Pi)$ on $\fF_{n_0}$.
    
    By definition, algorithm $\fA_{\det}$ is simply a function from $\fI$ to $\fO$.
    Hence, while technically $\fA_{\det}$ has been defined only for forests from $\fF_{n_0}$, we can execute $\fA_{\det}$ also on forests with an arbitrary number of nodes and it will still yield a correct output.
    We conclude that $\fA_{\det}$ constitutes a deterministic $0$-round algorithm solving $f^{T(n_0)}(\Pi)$ on $\fF$.
    
    Now, we can simply apply Lemma~\ref{lem:up} $T(n_0)$ times, starting with $f^{T(n_0)}(\Pi)$, and obtain that there exists a deterministic algorithm solving $\Pi$ on $\fF$ in $T(n_0)$ rounds (which implies the existence of a randomized algorithm solving $\Pi$ on $\fF$ in $T(n_0)$ rounds, using the same argument as in the beginning of the proof).
    As $n_0$ is a fixed positive integer (depending only on $\Delta$ and $\Pi$), it follows that $\Pi$ can be solved in constant time on $\fF$ (both deterministically and randomized).
\end{proof}

Now, by combining Theorem~\ref{thm:REmain} with Lemma~\ref{lem:treetoforest} and Lemma~\ref{lem:LCL_and_nodeedgeLCL_are_same} (which guarantees that Lemma~\ref{lem:treetoforest}, which is stated for node-edge-checkable LCL problems, can be applied), we obtain our main result as a corollary.

\begin{theorem}[Formal version of \cref{thm:main_informal}]\label{thm:treemain}
    Let $\Pi$ be an arbitrary LCL problem that has (deterministic or randomized) complexity $o(\log^* n)$ on $\fT$.
    Then $\Pi$ can be solved in constant time on $\fT$ (both deterministically and randomized).
\end{theorem}

\section{The \volume Model Gap}
\label{sec:volume}
In this section, we show that a deterministic or randomized \volume algorithm (\lca) with a probe complexity of $o(\log^* n)$ implies a deterministic \volume algorithm (\lca) with a probe complexity of $O(1)$. We first show the speed-up only for deterministic \volume algorithms and later discuss how to extend the result to the full generality.
\begin{theorem}
\label{thm:volume_speedup}
There does not exist an LCL with a deterministic \volume complexity between $\omega(1)$ and $o(\log^* n)$.
\end{theorem}
\begin{proof}
Consider some LCL $\Pi = (\Sigma_\inn, \Sigma_\out, r, \fP)$. Let $\mathcal{A}$  be a \volume model algorithm with a probe complexity of $T(n) = o(\log^*n)$ that solves $\Pi$. We show that this implies the existence of a \volume model algorithm that solves $\Pi$ and has a probe complexity of $O(1)$. To do so, we first show that we can turn $\mathcal{A}$ into an order-invariant algorithm with the same probe complexity. Once we have shown this, \cref{thm:speedup_order_invariant}, which shows a speed-up result for order-invariant algorithms, implies that there also exists a \volume model algorithm for $\Pi$ with a probe complexity of $O(1)$, as needed.

We start by proving the lemma below, which, informally speaking, states that there exists a sufficiently large set $S_n \subseteq [n]$ such that $\mathcal{A}$ is order-invariant as long as the IDs it "encounters" are from the set $S_n$. The proof adapts a Ramsey-theoretic argument first introduced in \cite{naorstockmeyer} and further refined in \cite{chang2016exp_separation} in the context of the \local model to the \volume model.

\begin{lemma}
\label{lem:order_invariant_over_set}
There exists a $n_0 \in \mathbb{N}$ such that the following holds for every $n \geq n_0$. There exists a set $S_n \subseteq [n]$ of size $(T(n) + 1) \cdot \Delta^{r+1}$ such that for every $i \in [T(n)+1]$ and $t,t' \in Tuples_{i,S_n}$ it holds that $f_{n,i}(t) = f_{n,i}(t')$ if $t$ and $t'$ are almost identical.
\end{lemma}
\begin{proof}
Consider an $n \in \mathbb{N}$. We denote with $H$ a complete $(T(n) + 1)$-uniform hypergraph on $n$ nodes. Each node in $H$ corresponds to a (unique) ID in the set $\{1,2,\ldots,n\}$.
For each hyperedge $X \subseteq \{1,2,\ldots,n\}$ we define a function $f_X$.

The input to $f_X$ is a tuple $t = (t_1,\ldots,t_{T(n)+1}) \in Tuples_{T(n)+1,[T(n) + 1]}$. 
For $j \in [T(n)+1]$, let $t_j = (id_j,deg_j,in_j)$. We define a new tuple $t^X_j = (id^X_j,deg_j,in_j)$, where $id^X_j$ is the $id_j$-th smallest element in $X$. 
We now define

\[f_X(t) = (f_{n,1}(t^X_1), f_{n,2}(t^X_1,t^X_2),\ldots,f_{n,T(n)+1}(t^X_1,t^X_2,\ldots,t^X_{T(n) + 1})).\]

We now prove two things. First, we show that for $n$ being sufficiently large, there exists a set $S^{big}_n \subseteq [n]$ of size $(T(n)+1)\cdot \Delta^{r+1} + T(n) + 1$ such that for any two hyperedges $X,Y \subseteq S^{big}_n$, $f_X = f_Y$. This will follow by Ramsey Theory and an upper bound on the number of possible different functions $f_X$.
Second, let $S_n$ be the set one obtains from $S^{big}_n$ by discarding the $T(n) + 1$ largest elements of $S^{big}_n$. We show that $S_n$ satisfies the conditions of \cref{lem:order_invariant_over_set}.

For the Ramsey-theoretic argument, we start by upper bounding the total number of possible functions. Note that $|Tuples_{[T(n) + 1]}| \leq (T(n)+1) \cdot \Delta \cdot |\Sigma_{in}|^{\Delta}$ and therefore the possible number of inputs to the function $f^X$ is $|Tuples_{T(n)+1,[T(n)+1]}| \leq \left( (T(n)+1) \cdot \Delta \cdot |\Sigma_{in}|^{\Delta} \right)^{T(n) + 1} = T(n)^{O(T(n))}$. Note that the output of $f_X$ is contained in the set $\left( \bigtimes_{i=1}^{T(n)} ([i] \times [\Delta]) \right) \times (\Sigma_{out})^{[\Delta]}$
and therefore there are at most $(T(n) \cdot \Delta )^{T(n)} \cdot |\Sigma_{out}|^{\Delta} = T(n)^{O(T(n))}$ possible outputs. Hence, there exist at most $\left(T(n)^{O(T(n))}\right)^{T(n)^{O(T(n))}}$ different possible functions. Let $R(p,m,c)$ denote the smallest number such that any $p$-uniform complete hypergraph on $R(p,m,c)$ nodes with each hyperedge being assigned one of $c$ colors contains a monochromatic clique of size $m$. It holds that $\log^*(R(p,m,c)) = p + \log^* m + \log^* c + O(1)$ \cite{chang2017time_hierarchy}.

Setting $p = T(n) + 1$, $m = (T(n) + 1) \cdot \Delta^{r+1} + T(n) + 1$ and $c = \left(T(n)^{O(T(n))}\right)^{T(n)^{O(T(n))}}$, we can conclude that $H$ contains a set $S_n^{big} \subseteq [n]$ of size $m$ such that for any two hyperedges $X,Y \subseteq S_n^{big}, f_X = f_Y$ as long as $\log^*(n) \geq T(n) + 1 + \log^* m + \log^* c + O(1)$, which is the case for sufficiently large $n$ as $T(n) = o(\log^* n)$, $\log^* m = o(T(n))$ and $\log^* c = o(T(n))$. \\

Now, let $S_n$ be the set one obtains form $S^{big}_n$ by discarding the $T(n) + 1$ largest elements from $S^{big}_n$.
Let $i \in [T(n) + 1]$ and $t^{(\ell)} = ((id^{\ell}_1,deg_1,in_1), $ $(id^{\ell}_2,deg_2,in_2),$ $ \ldots, (id^{\ell}_i,deg_i,in_i)) \in Tuples_{i,S_n}$ for $\ell \in [2]$ such that $t^{(1)}$ and $t^{(2)}$ are almost identical. It remains to show that $f_{n,i}(t^{(1)}) = f_{n,i}(t^{(2)})$. For $\ell \in [2]$, let $X^{\ell} \subseteq S^{big}_n$ such that $\{id^\ell_1,id^\ell_2,\ldots,id^\ell_i \}$ contains the $|\{id^\ell_1,id^\ell_2,\ldots,id^\ell_i \}|$-th lowest elements of $X^{\ell}$. 
Now, let $t = (t_1,t_2,\ldots,t_{T(n) + 1}) \in S_{T(n) + 1,[T(n) + 1]}$ such that for every $j \in [i]$, $t^{X^\ell}_j = (id^{\ell}_j,deg_j,in_j)$ . Note that it follows from the way we defined $X^1$ and $X^2$ and the fact that $t^{(1)}$ and $t^{(2)}$ are almost identical that such a tuple $t$ exists.
As $X^{\ell} \subseteq S^{big}_n$, we have $f_{X^1}(t) = f_{X^2}(t)$. In particular, this implies that \[f_{n,i}(t^{(1)}) = f_{n,i}(t^{X^1}_1,t^{X^1}_2, \ldots, t^{X^1}_i) = f_{n,i}(t^{X^2}_1,t^{X^2}_2, \ldots, t^{X^2}_i) = f_{n,i}(t^{(2)}),\]

which finishes the proof.

\end{proof}

We now construct an order-invariant algorithm $\fA'$ with probe complexity $T'(n) = \max(O(1),T(n)) = o(\log^*n)$. Note that it is easy to make $\fA'$ order-invariant for every input graph having fewer than $n_0$ nodes. For $n \geq n_0$, we have $T'(n) = T(n)$ and for every $i \in \{1,2,\ldots,T'(n)\}$ and tuple $t = ((id_1,deg_1,in_1),(id_2,deg_2,in_2),\ldots,(id_i,deg_i,in_i)) \in Tuples_{i,\mathbb{N}}$, we define

\[f^{\fA'}_{n,i}(t) = f^{\fA}_{n,i}(t'=(id'_1,deg_1,in_1),(id'_2,deg_2,in_2),\ldots,(id'_i,deg_i,in_i))\]

where $id'_1,\ldots,id'_i$ is chosen in such a way that $\{id'_1,\ldots,id'_i\}$ contains the $|\{id'_1,\ldots,id'_i\}|$-th smallest elements of $S_n$ and $t$ and $t'$ are almost identical.
It is easy to verify that $\fA'$ is indeed order-invariant. 

It remains to show that $\fA'$ indeed solves $\Pi$.
For the sake of contradiction, assume that this is not the case. 
This implies the existence of a $\Sigma_{in}$-labeled graph $(G,f_{in})$ on $n$ nodes (with each node having a unique ID from a polynomial range, a port assignment, and $G$ does not have any isolated nodes) such that $\fA'$ "fails" on $(G,f_{in})$.
Put differently, there exists a node $v$ such that $\fA'$ produces a mistake in the $r$-hop neighborhood of $v$.
The $r$-hop neighborhood of $v$ consists of at most $\Delta^{r+1}$ vertices. To answer a given query, $\fA'$ "sees" at most $T(n) + 1$ nodes. Hence, to compute the output of all the half-edges in the $r$-hop neighborhood of $v$, $\fA'$ "sees" at most  $\Delta^{r+1}(T(n) + 1) \leq |S_n|$ many nodes. We denote this set of nodes by $V^{visible}$.
Even if the IDs of nodes outside of $V^{visible}$ are changed, $\fA'$ still fails around $v$. Moreover, as $\fA'$ is order-invariant, changing the IDs of the nodes in $V^{visible}$ in a manner that preserves the relative order does not change the fact that $\fA'$ fails around $v$. Hence, we can find a new assignment of IDs such that each node in $V^{visible}$ is assigned an ID from the set $S_n$ such that $\fA'$ still fails around $v$. However, from the way we defined $\fA'$ and the property that $S_n$ satisfies, it follows that $\fA'$ and $\fA$ produce the same output in the $r$-hop neighborhood around $v$. This contradicts the fact that $\fA$ is a correct algorithm.

\end{proof}

We already argued in the preliminaries that \cref{thm:volume_speedup} also implies that there does not exist an LCL with a deterministic \lca complexity between $\omega(1)$ and $o(\log^* n)$. 

As noted in \cite{Rosenbaum2020}, the derandomization result by \cite{chang2016exp_separation} can be used to show that randomness does not help (up to an additive constant in the round/probe complexity) in the \local and \volume model for complexities in $O(\log^* n)$, and the same is true for the \lca model.

Hence, we obtain the following more general theorem.

\begin{theorem}
There does not exist an LCL with a randomized/deterministic \lca/\volume complexity between $\omega(1)$ and $o(\log^* n)$.
\end{theorem}

\section{Speedup in Grids}
\label{sec:grids}

In this section we prove that on $d$-dimensional oriented grids, any $o(\log^* n)$-round \local algorithm for an LCL problem $\Pi$ (with input labels) can be turned into a \local algorithm for $\Pi$ with a round complexity of $O(1)$. An oriented $d$-dimensional grid is a grid where each edge is labeled with a label from $[d]$ that says which dimension this edge corresponds to. Moreover, all edges corresponding to the same dimension are consistently oriented. The grid is assumed to be toroidal, that is, without boundaries. We believe, but do not prove, that the same result can be shown also for non-toroidal grids, that is, grids with boundaries. From now on fix a dimension $d$. 

\begin{theorem}
\label{thm:main_grids}
There is no LCL problem $\Pi$ with a randomized/deterministic \local complexity between $\omega(1)$ and $o(\log^* n)$ in $d$-dimensional oriented grids for $d = O(1)$.
\end{theorem}

The proof of the speedup is a relatively simple extension of an argument of Suomela in the paper of Chang and Pettie \cite{chang2016exp_separation}. There, the argument is proven for LCLs without inputs on unoriented grids. It shows that every $o(\log^* n)$-round algorithm can be even sped up to a $0$-round algorithm. This is not the case anymore with inputs, which makes the proof slightly more involved.

The high-level idea of the proof is the following. First, we observe that a \local algorithm on oriented grids implies a \prodlocal algorithm with the same asymptotic complexity. A \prodlocal algorithm is defined as follows. 

\begin{definition}[\prodlocal model]
An algorithm in the \prodlocal model is defined as an algorithm in the \local model, but it expects that each node $u \in V(G)$ for $G$ an oriented grid gets an ordered sequence of $d$ identifiers $id_1(u), \dots, id_d(u) \in [n^{O(1)}]$, one for each dimension of the input grid $G$. 
We require that the $i$-th identifier $id_i(u), id_i(v)$ of two nodes $u,v$ is the same if and only if the two nodes have the same $i$-th coordinate in $G$.

An order-invariant \prodlocal model is defined as follows. 
We say that two identifier assignments are order-indistinguishable in some neighborhood around a node $u$ if for any two nodes $v,w$ in that neighborhood of $u$ and any $i,j \in [d]$, the $i$-th identifier of one node is either smaller than the $j$-th identifier of the other node in both ID assignments or larger in both ID assignments. Then, an order-invariant $t(n)$-round \prodlocal algorithm outputs the same value at the half edges of a node in case two ID assignments are order-indistinguishable in its $t(n)$-hop neighborhood. 
\end{definition}

We have the following fact:
\begin{proposition}
\label{fact:prod}
If an LCL $\Pi$ allows has local deterministic/randomized complexity $t(n)$ in the \local model, it has complexity $t(n)$ in the \prodlocal model. 
\end{proposition}

\begin{proof}
By result of \cite{chang2016exp_separation}, a $t(n)$ round randomized algorithm in the local model implies a $t(2^{O(n^2)})$ round deterministic algorithm. This implies that we can turn an $O(\log^*n)$-rounds randomized algorithm into $O(\log^*n)$-rounds deterministic one. 

Each of the $d$ identifiers of a node are bounded by $n^c$ for some fixed constant $c$ in the \prodlocal model. Assigning each node the identifier $I = \sum_{i=1}^d I_i \cdot n^{c(i-1)}$, where $(I_i)$ is its $i$-th identifier results in globally unique identifiers from a polynomial range. This allows simulation of the deterministic \local algorithm in the \prodlocal model as needed. 
\end{proof}

Next, we argue that an $o(\log^* n)$-round \prodlocal algorithm can be turned into an order-invariant \prodlocal algorithm. 
To show the existence of such an order-invariant algorithm, we use a Ramsey theory based argument very similar to the one in \cref{sec:volume} and \cite{chang2016exp_separation}. 

\begin{proposition}
\label{prop:ramsey}
If there is a $t(n) = o(\log^* n)$ round \prodlocal algorithm for an LCL $\Pi$, then there is also an order-invariant \prodlocal algorithm with the same round complexity.
\end{proposition}
\begin{proof}
The algorithm $\fA$ of local complexity $t(n)$ can be seen as a function that maps $(2t(n)+1) \times (2t(n)+1)$ sized grids with edges consistently oriented and labeled with numbers from $[d]$. Moreover, each vertex has a $d$-tuple of its identifiers and each half-edge a label from $\Sigma_{in}$. The function $\fA$ maps this input to an output label for each half-edge incident to a given node.  
Let $R(p,m,c)$ be defined as in the previous section and let $H$ be a hypergraph on $n^{O(1)}$ nodes, each representing an identifier. 

\begin{enumerate}
    \item Assume we have a $t(n)$-hop neighborhood already labeled without assigned identifiers and we want to assign the identifiers. There are at most $p = d \cdot (2t(n)+1)$ numbers to assign. 
    \item Assume we have a $t(n) + r$-hop neighborhood already labeled without assigned identifiers and we want to assign the identifiers. There are at most $m = d \cdot (2(t(n) + r) + 1)$ numbers to assign. Here, $r$ is the local checkability radius of the problem. 
    \item The number $z$ is defined to be the number of possible input neighborhoods not labeled with identifiers. We have $z \leq |\Sigma_{in}|^{2d(2t(n)+1)^d}$. Note that the extra $2d$-factor in the exponent comes from the fact that we label half-edges. 
    \item Finally, $c$ is the number of colors such that each color encodes a distinct function which takes as input one of the $z$ possible neighborhoods labeled with $\Sigma_{in}$ inputs. Moreover, for each such neighborhood we fix a given permutation on $p$ elements $\pi$. The function outputs a tuple of $d$ output labels from $\Sigma_{out}$, one entry corresponding to each half-edge incident to a vertex. 
    Hence, for the number of such colors $c$ we have $c = |\Sigma_{out}|^{d p! z}$. 
\end{enumerate}

We color each $p$-edge $T$ of $H$ by the following color out of $c$ colors: we consider all of the $z$ possible ways of labeling the $t(n)$-hop neighborhood with labels from $\Sigma_{in}$ and all $p!$ ways of how one can assign a set of $p$ different identifiers of $T$ to that neighborhood (we think of $T$ as a sorted list on which we apply $\pi$ and then in this order we fill in identifiers to the $t(n)$-hop neighborhood). For each such input (out of $p!\cdot z$ possible) we now apply $\fA$ to that neighborhood and record the output from $\Sigma_{out}^d$. This gives one of $c$ possible colors.  
We now use the bound $\log^* R(p,m,c) = p + \log^*m + \log^*c + O(1)$ as in the previous section. 
Hence, as we assume $t(n) = o(\log^* n)$, this implies $|V(H)| \ge R(p,m,c)$, and therefore there exists a set $S \subseteq V(H)$ of $m$ distinct IDs such that all $p$-sized subsets $T \subseteq S$ have the same color. 

We now define the new algorithm $\fA'$ as follows. Let $N$ be a neighborhood from the order-invariant \prodlocal model; 
take an arbitrary $p$-sized subset $T\subseteq S$ and label $N$ with identifiers from $T$ by the permutation $\pi$ given by the order on $N$ that we got. Then, apply $\fA$ on this input. 

First, this algorithm is well defined as it does not matter which $T$ we pick, the algorithm always answers the same because every $T \subseteq S$ is colored by the same color. 
Second, the algorithm is correct, since for any $(t(n)+r)$-hop neighborhood that has at most $m$ vertices we can choose some way of assigning identifiers from $S$ ($S = m$) to that neighbourhood and $\fA'$ answers the same as what $\fA$ would answer with this assignment of identifiers for all nodes within the local checkability radius. 
\end{proof}

Once we have the order-invariant algorithm with a round complexity of $o(\log^* n)$, an easy adaptation of the proof for \cref{thm:speedup_order_invariant} implies that we can turn it into an order-invariant \prodlocal algorithm with a round complexity of $O(1)$. 
However, because of the oriented-grid, we get a local order on the vertices for free, so we can turn the order-invariant \prodlocal algorithm into an (order-invariant) \local algorithm, thus finishing the proof. 

\begin{proposition}
\label{prop:final}
If an LCL $\Pi$ has complexity $o(n^{1/d})$ in the order-invariant \prodlocal model, it has deterministic/randomized/deterministic order-invariant local complexity $O(1)$ in the \local model. 
\end{proposition}
\begin{proof}
This is an argument very similar to  \cite{chang2017time_hierarchy}. 
Let $\fA$ be the algorithm in the \prodlocal model. 
We choose $n_0$ large enough and run $\fA$ ``fooled'' into thinking that the size of the grid is $n_0$. 
As the order on the identifiers, choose the order where two identifiers $id_i(u), id_j(v)$ have $id_i(u) < id_j(v)$ if and only if either $i < j$ or $i = j$ and $v$ is further than $u$ in the $i$-th coordinate (the notion of ``further than'' is given by the orientation of the grid). 
A standard argument as in \cite{chang2017time_hierarchy} shows that if $n_0$ is chosen large enough, the ``fooled'' algorithm needs to be correct, as otherwise $\fA$ would be incorrect when being run on grids of size $n_0$. 
\end{proof}

\cref{thm:main_grids} follows by application of \cref{fact:prod,prop:ramsey,prop:final}.

\section*{Acknowledgements}
We would like to thank Yi-Jun Chang, Mohsen Ghaffari, Jan Grebík, and Dennis Olivetti for helpful comments and discussions.

CG and VR were supported by the European Research Council (ERC) under the European Unions Horizon 2020 research and innovation programme (grant agreement No. 853109).

\bibliography{ref}

\appendix

\end{document}